\documentclass[a4paper,UKenglish,cleveref, autoref, thm-restate]{lipics-v2021}

\hideLIPIcs  

\title{Approximating $\delta$-Covering}

\author{Tim A. Hartmann}{CISPA Helmholtz Center for Information Security, Germany}{
tim.hartmann@cispa.de}{https://orcid.org/0000-0002-1028-6351}{}

\author{Tom Jan{\ss}en}{Department of Computer Science, RWTH Aachen University, Germany}{janssen@algo.rwth-aachen.de}{https://orcid.org/0000-0003-4617-3540}{Deutsche Forschungsgemeinschaft (DFG) –- WO 1451/2-1}

\authorrunning{T.\ A.\ Hartmann and T.\ Jan{\ss}en} 

\Copyright{Tim A.\ Hartmann and Tom Jan{\ss}en} 


\keywords{Facility Location, Approximation Algorithms, Dominating Set, Vertex Cover}

\category{} 

\relatedversion{} 

\nolinenumbers 

\EventEditors{John Q. Open and Joan R. Access}
\EventNoEds{2}
\EventLongTitle{42nd Conference on Very Important Topics (CVIT 2016)}
\EventShortTitle{CVIT 2016}
\EventAcronym{CVIT}
\EventYear{2016}
\EventDate{December 24--27, 2016}
\EventLocation{Little Whinging, United Kingdom}
\EventLogo{}
\SeriesVolume{42}
\ArticleNo{23}

\usepackage{setspace}
\usepackage{pgfplots}
\usetikzlibrary{shapes.misc}
\usepackage{amsmath,amssymb,amsfonts}
\usepackage{dsfont}
\usepackage{enumerate}
\usepackage[square,sort,comma,numbers]{natbib}
\usepackage{gensymb}
\usepackage{array}
% ...
\usepackage{mathtools}
\usepackage{listings}
\usepackage[ruled,vlined,linesnumbered]{algorithm2e}
\usepackage{color}
\usepackage{xcolor}

\makeatletter
\newcommand{\labeltext}[3][]{%
    \@bsphack%
    \csname phantomsection\endcsname
    \def\tst{#1}%
    \def\labelmarkup{}
    \def\refmarkup{}%
    \ifx\tst\empty\def\@currentlabel{\refmarkup{#2}}{\label{#3}}%
    \else\def\@currentlabel{\refmarkup{#1}}{\label{#3}}\fi%
    \@esphack%
    \labelmarkup{#2}
}
\makeatother

\newcommand{\cov}[1][\delta]{\ensuremath{#1\text{-\textsc{Covering}}}}

\newcommand{\vc}{\textsc{Vertex Cover}}

\newcommand{\OPT}{\operatorname{OPT}}

\newcommand{\NP}{{NP}} 
\newcommand{\NPO}{{NPO}} 
\newcommand{\APX}{{APX}} 
\newcommand{\logAPX}{{log-APX}} 
\newcommand{\UG}{\textsl{UG}} 

\newcommand{\dist}[1]{\ensuremath{d\left(#1\right)}}
\newcommand{\wreath}{\ensuremath{\textup{ wr }}}
\newcommand{\covn}[2][\delta]{\ensuremath{#1\textup{-cover}(#2)}}
\newcommand{\dispn}[2][\delta]{\ensuremath{#1\textup{-disp}(#2)}}

\newcommand{\half}{\tfrac{1}{2}}

\AddToHook{cmd/appendix/after}{}

\newcommand{\ball}[2]{B^{<}(#1,#2)}
\newcommand{\balll}[2]{B^{\leq}(#1,#2)}

\newcommand{\CC}{\mathcal{C}}

\newcommand{\GG}{\mathcal{G}}

\definecolor{dkgreen}{rgb}{0,0.6,0}
\definecolor{dkblue}{rgb}{0.1,0,0.6}
\definecolor{gray}{rgb}{0.5,0.5,0.5}
\definecolor{lgray}{rgb}{0.85,0.85,0.85}
\definecolor{mauve}{rgb}{0.58,0,0.82}
\definecolor{red}{rgb}{1.0,0,0}
\definecolor{lorange}{rgb}{1.0,0.7,0}
\definecolor{lgreen}{rgb}{0,0.8,0}
\hideLIPIcs

\begin{document}

\maketitle

\bibliographystyle{unsrt}

\begin{abstract}
$\delta$-\textsc{Covering}, for some covering range $\delta>0$, is a continuous facility location problem on 
undirected graphs
	where all edges have unit length.
The facilities may be positioned on the vertices as well as on the interior of the edges.
The goal is to position as few facilities as possible
	such that every point on every edge
	has distance at most $\delta$ to one of these facilities.
For large $\delta$, the problem is similar to dominating set, 
	which is hard to approximate,
	while for small $\delta$, say close to $1$, the problem is similar to vertex cover.
In fact, as shown by Hartmann et al.~[Math.\ Program.~22],
	$\delta$-Covering for all unit-fractions $\delta$
	is polynomial time solvable,
	while for all other values of $\delta$ the problem is NP-hard.

We study the approximability of $\delta$-Covering for every covering range $\delta>0$.
For $\delta \geq 3/2$,
	the problem is \logAPX-hard,
	and allows an $\mathcal O(\log n)$ approximation.
For every $\delta < 3/2$,
	there is a constant factor approximation of a minimum $\delta$-cover
	(and the problem is \APX-hard when $\delta$ is not a unit-fraction).
We further study the dependency of the approximation ratio on the covering range $\delta < 3/2$.
By providing several polynomial time approximation algorithms
	and lower bounds under the Unique Games Conjecture,
	we narrow the possible approximation ratio,
	especially for $\delta$ close to the polynomial time solvable cases.
\end{abstract}

\vfill
\newpage
\section{Introduction}

We study the approximability of a continuous facility location problem.
The input is a graph $G$ whose edges have unit length.
Let $P(G)$ be the continuum set of points on all the edges and vertices.
The distance of two points $p,q \in P(G)$
	is the length of a shortest path connecting $p$ and $q$ in the underlying metric space.
For a rational $\delta>0$, a subset of points $S \subseteq P(G)$ is a \emph{$\delta$-cover}
	if every point $p \in P(G)$ has distance $d(p,q)\leq \delta$ to some point $q \in S$.
For a fixed $\delta>0$, and given a graph $G$,
	our objective is to compute a minimum size $\delta$-cover of $G$.
We denote the minimum cardinality of a $\delta$-cover $S$ as $\covn[\delta]{G}=|S|$.
The decision-version of this problem is known as $\cov[\delta]$.

The computational complexity of computing an optimal $\delta$-cover is relatively well understood.
$\cov[\delta]$ is polynomial time solvable for every covering range $\delta$
	that is a unit-fraction, 
	while it is \NP-hard for all other $\delta$,
	as shown by Hartmann et al.~\cite{DBLP:journals/mp/HartmannLW22}.
The same work also settles the parameterized complexity with the solution size as parameter.
When $\delta < 3/2$, the problem is fixed-parameter tractable,
	while for $\delta \geq 3/2$, the problem is W[2]-hard.
The approximability depending on $\delta$ has not yet been settled.

As the problem is polynomial time solvable when $\delta$ is a unit-fraction,
	we are particularly interested in the approximability when $\delta$ is near a unit-fraction.
Also, as revealed by the parameterized complexity study,
	for large $\delta$, the problem is similar to the dominating set problem 
	-- in the sense that the main focus is on covering the vertices, 
	which is hard to approximate.
On the other hand, for small $\delta$, say close to $1$, 
	the problem is similar to the vertex cover problem 
	-- in the sense that the main focus is on covering the edges,
	which is well-known to allow a $2$-approximation.
Similarly, for the approximation,
	we expect some threshold for $\delta$,
	till which the problem allows a constant factor approximation (like vertex cover),
	while beyond this threshold the problem is \logAPX-hard (like dominating set).

\subparagraph{Further Related Results}
Facility Location in general is a wide research area.
We refer to the books by Drezner~\cite{drezner1996facility}
	and Mirchandani \& Francis~\cite{mirchandani1990discrete}.
Our model, where the metric space is defined by a graph, dates back to Dearing \& Francis 
\cite{Dearing1974}.
Several optimization goals have been considered, for example by works of 
Tamir~\cite{Tamir87,Tamir1991}.
For \cov, Megiddo \& Tamir~\cite{MegiddoT1983} gave a polynomial time algorithm for trees,
	and show \NP-hardness for $\delta=2$.

A dual problem to \cov, is known as $\delta$-\textsc{Dispersion}.
The task is to select a maximum number of points that have distance at least $\delta$.
Similarly to \cov, the problem is \NP-hard for all $\delta$ that 
	are not a unit fraction or twice a unit fraction,
	as shown by Grigoriev et al.~\cite{GrigorievHLW21}.
The problem is fixed-parameter tractable in the solution size for $\delta \leq 2$,
	and W[1]-hard otherwise, as shown by Hartmann et al.~\cite{DBLP:conf/mfcs/HartmannL22}.
The same work addresses the parameterized complexity
	with respect to various graph parameters.
A similar task was studied when approximating the problem:
Instead of fixing the minimum distance, the number of points to place is fixed,
	and the goal is to maximize the minimum pairwise distance between the placed points.
For this task, Tamir~\cite{Tamir1991} gave a $\half$-approximation algorithm,
	while they showed that there is no $\varepsilon$-approximation for $\varepsilon > 
	\frac{2}{3}$,
	unless $\textsc{P}=\textsc{NP}$.

\subparagraph{Our Results}
For $\delta \geq 3/2$,
	\cov{} is \logAPX-hard,
	and allows an $\mathcal O(\log n)$ approximation.
For every $\delta < 3/2$,
	the problem has a constant factor approximation
	and is \APX-hard when $\delta$ is not a unit-fraction.

For a covering range $\delta < 3/2$, we give several polynomial time approximation algorithms
	and lower bounds under the Unique Games Conjecture (UGC),
	plotted in part in \Cref{fig:ug-bounds,fig:ug-bounds:high}.
Due to a result by Hartmann et al.~\cite{DBLP:journals/mp/HartmannLW22}	that lets us translate the value of $\delta$,
	we may mostly focus on $\delta \geq 1/2$.
One may expect that the approximation ratio approaches $1$
	if $\delta$ approaches
	a polynomial time solvable case
	such as $\delta=1$ and $\delta=1/2$.
\begin{itemize}
\item 
This is true for when $\delta$ approaches $1/2$ from above.
We give families of upper and lower bounds
	which in the limit towards $1/2$ reach $1$.
That is, for $\delta \in [ \frac{x+1}{2x+1}, \frac{x}{2x-1} )$ for integer $x \geq 2$,
	we provide a $\frac{x+1}{x}$-approximation algorithm
	(\Cref{lemma:approx:1:2:to:2:3}).
Our lower bounds under UGC are provided by three families of lower bounds
	that together cover every $\delta$ close to $1/2$.
For example a lower bound of $1 + \frac{1}{2x+1}$
	for covering range $\delta \in [\frac{x+1}{2x+1}, \frac{2x+1}{4x})$ for integer $x \geq 1$
	(\Cref{theorem:uglbs}).
\item
On the other hand, this is not true when $\delta$ approaches $1$ from below.
We give an approximation lower bound under UGC of $6/5$  for $\delta \in [\frac{5}{6},1)$
	(also \Cref{theorem:uglbs}).
Further, we provide a $2$-approximation algorithm for a superset of these $\delta$
	(\Cref{lemma:bound:delta:smaller:1}).
\end{itemize}
Further, for the interval $\delta \in [\frac{2}{3}, \frac{3}{4})$
	we provide an upper and lower bound that only differ by a factor of $9/8$.
Our upper bound is an algorithm
	that behaves uniformly for every $\delta \in [\frac{2}{3}, \frac{3}{4})$,
	while the minimum size of a $\delta$-cover for such $\delta$ can vary by a factor of ${9}/{8}$ itself. 

\medskip

For covering range $\delta \in (1,\frac{3}{2})$, the situation is similar to $\delta \in (\half,1)$.
We give a $2$-approximation for this interval,
	while for $\delta < {5}/{4}$ we reduce this factor to $5/3$
	and for $\delta < 7/6$ we reduce this factor to $3/2$.
The lower bounds are inherited from the case $\delta \in (\frac{1}{2},\frac{3}{4})$.

For a covering range $\delta < \frac{1}{2}$,
	we inherit upper bounds from the case $\delta \in (\frac{1}{2},1)$,
	which we further improve for certain ranges of $\delta$.

\subparagraph{Organization of this Work}
Preliminaries are in \Cref{section:preliminaries}.
\Cref{section:apx} settles the \APX-hardness and case $\delta \geq \frac{3}{2}$,
	while \Cref{section:hardness:ugc} provides hardness results under UGC.
We give our upper bounds for $\delta > 1$, $\delta \in (\half,1)$ and $\delta < \half$
	in \Cref{section:approx:large}, \Cref{section:simple:approx}, \Cref{section:approx:small}, respectively.

\newcommand{\putmarkx}[1]{
	\pgfmathsetmacro{\xval}{(\x-0.5)*20}
	\draw (\xval,0) -- (\xval,-0.15);
	\node (putmarkx) at (\xval, -0.5) {$#1$};
}
\newcommand{\putmarky}[1]{
	\pgfmathsetmacro{\yval}{(\y-1)*5}
	\draw (-0.2, \yval) -- (0, \yval);
	\node (putmarky) at (-0.5+0.1*\c, \yval) {$#1$};
}

\begin{figure}[t]
	\centering
	\begin{tikzpicture}[yscale=0.7,xscale=1.1]
		\def\colora{dkblue}
		\def\colorb{lgreen}
		\def\colorc{blue}
		\def\colord{cyan}
		\def\colore{purple}
		\def\colorf{lorange}
		\def\colorg{red}
		\draw[->, ultra thick] (-0.2, 0) -- (10.5, 0);
		\node at (11, -0.5) {$\delta$};
		\draw[->, ultra thick] (0, -0.2) -- (0, 5.5);
		\node[rotate=90] at (-1, 2) {ratio};
		
		\def\x{1}\putmarkx{1}
		\foreach \a/\b in {1/2,3/4,2/3,4/5,3/5,4/7,5/8,9/14,7/10,5/6}{
			\def\x{\a/\b}\putmarkx{\frac{\a}{\b}}
		}
		
		\def\c{0}\def\y{2}\putmarky{2}
		\foreach \a/\b/\c in {8/7/1, 6/5/-1, 4/3/1, 3/2/0} {
			\def\y{\a/\b}\putmarky{\frac{\a}{\b}}
		}
		\def\y{1}\putmarky{1}

		\foreach \x in {1, ..., 100} {
			\pgfmathsetmacro{\from}{((\x+1)/(2*\x+1)-0.5)*20}
			\pgfmathsetmacro{\to}{((2*\x+1)/(4*\x)-0.5)*20}
			\pgfmathsetmacro{\val}{((2*\x+2)/(2*\x+1)-1)*5}
			\ifthenelse{\x=1}{
				\draw[-, color=\colora] (\from, \val) -- node[above, yshift=-3] {Th.~\ref{theorem:uglbs}(a)} (\to, \val);
			}{
				\draw[-, color=\colora] (\from, \val) -- (\to, \val);
			}
			
			\pgfmathsetmacro{\from}{((\x+2)/(2*\x+2)-0.5)*20}
			\pgfmathsetmacro{\to}{((2*\x+3)/(4*\x+2)-0.5)*20}
			\pgfmathsetmacro{\val}{((2*\x+2)/(2*\x+1)-1)*5}
			\ifthenelse{\x=1}{
				\draw[-, color=\colorb] (\from, \val) -- node[above, yshift=-3] {Th.~\ref{theorem:uglbs}(b)} (\to, \val);
			}{
				\draw[-, color=\colorb] (\from, \val) -- (\to, \val);
			}
			
			\pgfmathsetmacro{\from}{((2*\x+5)/(4*\x+6)-0.5)*20}
			\pgfmathsetmacro{\to}{((\x+2)/(2*\x+2)-0.5)*20}
			\pgfmathsetmacro{\val}{((2*\x+6)/(2*\x+5)-1)*5}
			\ifthenelse{\x=1}{
				\draw[-, color=\colorc] (\from, \val) -- node[above, yshift=-3, xshift=5] {Th.~\ref{theorem:uglbs}(c)} (\to, \val);
			}{
				\draw[-, color=\colorc] (\from, \val) -- (\to, \val);
			}
			
			\pgfmathsetmacro{\from}{((\x+2)/(2*(\x+1)+1)-0.5)*20}
			\pgfmathsetmacro{\to}{((\x+1)/(2*(\x+1)-1)-0.5)*20}
			\pgfmathsetmacro{\val}{((\x+2)/(\x+1)-1)*5}
			\ifthenelse{\x=1}{
				\draw[-, color=\colore, ultra thick] (\from, \val) -- node[above] {Th.~\ref{lemma:approx:1:2:to:2:3}} (\to, \val); 
			}{
				\draw[-, color=\colore, ultra thick] (\from, \val) -- (\to, \val);
			}
 		}
 		\pgfmathsetmacro{\x}{1}
		\pgfmathsetmacro{\from}{((\x+1)/(2*\x+1)-0.5)*20}
		\pgfmathsetmacro{\to}{((2*\x+1)/(4*\x)-0.5)*20}
		\pgfmathsetmacro{\val}{((2*\x+2)/(2*\x+1)-1)*5}
		\draw[-, color=\colorf, ultra thick] (\from, 2.5) -- node[above] {Th.~\ref{lemma:approx:3:2:to:3:4}} (\to, 2.5);
 		
 		\pgfmathsetmacro{\x}{1}
		\pgfmathsetmacro{\from}{((\x+2)/(2*\x+2)-0.5)*20}
		\pgfmathsetmacro{\to}{((\x+1)/(2*\x)-0.5)*20}
		\pgfmathsetmacro{\val}{((\x+1)/(\x)-1)*5}
		\draw[-, color=\colorg, ultra thick] (\from, \val)
			-- node[above] {Th.~\ref{lemma:bound:delta:smaller:1}} (\to, \val);
 		
		\draw[-, color=\colord] (6, 1) -- (7.5, 1);
		\draw[-, color=\colord] (7.5, 1) -- node[above, yshift=-3] {Th.~\ref{theorem:uglbs}(d)} (10, 1);
	\end{tikzpicture}
	\caption{Upper bounds (as bold lines) and lower bounds under UGC (as thin lines)
	on the approximation ratio of \cov\ plotted for $\delta \in (\half,1)$.
	The drawn intervals are half-open with the upper end excluded.}
	\label{fig:ug-bounds}
\end{figure}
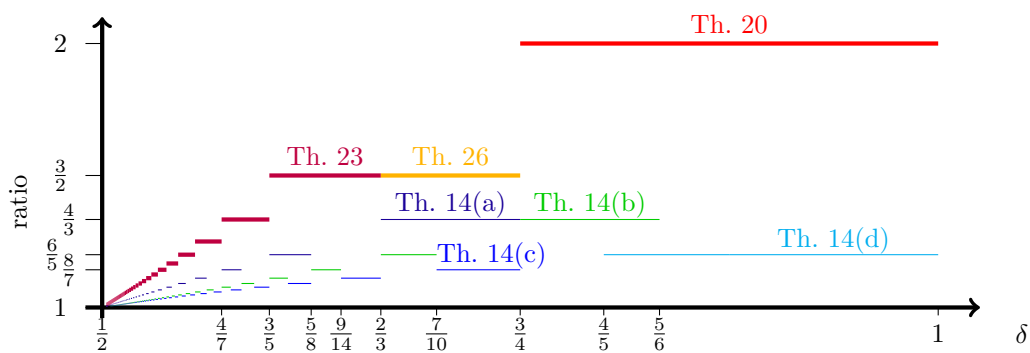

\section{Preliminaries}
\label{section:preliminaries}

All graphs we consider are simple and undirected.
Every edge has unit length.
We use $V(G)$ and $E(G)$ for the vertex respectively edge set of a graph $G$.
Further, $N(v) \coloneqq \{ u \in V(G) \mid \{u,v\} \in E(G) \}$ and $N[v] \coloneqq N(v) \cup \{v\}$.
We use the word \emph{vertex} in the graph-theoretic sense,
	while we use the word \emph{point} for the elements of $P(G)$.
We use the notation of~\cite{GrigorievHLW21}:
For an edge $\{u,v\}$ and real $\lambda \in [0,1]$,
	we denote by $p(u,v,\lambda)$ the point on edge $\{u,v\}$
	that has distance $\lambda$ to $u$.
We have that $p(u,v,\lambda)$ coincides with $p(u,v,1-\lambda)$,
	and that $p(u,v,0)=u$ and $p(u,v,1)=v$.
The \emph{open ball} with radius $r$ around a point $p$, denoted as $\ball{p}{r}$,
is the set of points with distance $< r$ from $p$.
The \emph{closed ball}, denoted as $\balll{p}{r}$, is the set of points points with distance $\leq r$ from $p$.

\begin{lemma}[Hartmann et al.~\cite{DBLP:journals/mp/HartmannLW22}]
	\label{lemma:covering:p}
	Let $\delta$ be a unit fraction.
	Then a minimum $\delta$-cover of a graph $G$ can be found in polynomial time.
\end{lemma}

The proof in~\cite{DBLP:journals/mp/HartmannLW22}
	uses the Edmonds-Gallai-Decomposition, see~\cite{Edmonds1965,Gallai1963,Gallai1964},
	to split $G$ into three disjoint subgraphs $G_0$, $G_1$ and $G_{\geq 3}$.
Graph $G_0$ has a perfect matching, $G_1$ is bipartite,
	and every component $C$ of $G_{\geq 3}$ has size at least $3$
	and is \emph{factor-critical},
	that is for any vertex $u \in V(C)$ there is a maximum matching that misses exactly $u$.
Let $c_{\geq 3}$ be the number of components of $G_{\geq 3}$.
Since every component contains at least three vertices, we have $c_{\geq 3} \leq \frac{1}{3}|V(G)|$.
Then the computed minimum $1$-cover $S^*$ has the size bound
\begin{equation}
\label{equation:1:cover:size}
	|S^*| \;=\; \nu(G_0) + \nu(G_{\geq 3}) + c_{\geq 3} + \text{vc}(G_1)
	\;\leq\; \tfrac{1}{2}(|V(G)| + c_{\geq 3})
	\;\leq\; \tfrac{2}{3}|V(G)|,
\end{equation}
where $\nu$ denotes the size of a maximum matching.
Further, we use the following property.
A subset $S \subseteq P(G)$ is \emph{neat} if
\begin{itemize}
\item \labeltext{$($N1$)$}{def:edge}
for every $\{u,v\}\in E(G)$, if $|S \cap P(G[\{u,v\}])| \geq 2$,
	then $S \cap P(G[\{u,v\}]) = \{u,v\}$. 
\end{itemize}

\begin{restatable}{lemma}{lemmaNOne}
\label{lemma:n:1}
Let $\delta\geq\half$ and $G$ be a graph.
There is a minimum $\delta$-cover of $G$ that is neat.
\end{restatable}
\begin{proof}
	If there is an edge $\{u,v\}$, where $|S \cap P(G[\{u,v\}])| \geq 2$
	but $S \cap P(G[\{u,v\}]) \neq \{u,v\}$,
	we replace the points $S \cap P(G[\{u,v\}])$ in $S$ by the two points $u,v$.
	Applying this modification for every edge yields a minimum $\delta$-cover $S$ that satisfies \ref{def:edge}.
\end{proof}

To obtain an $\alpha$-approximation of \covn{G},
	it suffices to compute an $\alpha$-approximation for every connected component of $G$.
Further, $\delta$-\textsc{Covering} is polynomial time solvable on trees,
	as shown by Megiddo et al.~\cite{MegiddoT1983}.
Hence, it suffices to consider connected non-tree graphs as input.

\subparagraph{Approximation Preserving Reductions}
The class \NPO{} consists of all \NP{} optimization problems.
The class $\APX{} \subseteq \NPO{}$ are those \NP{} optimization problems that admit a constant factor approximation in polynomial time.
A problem $\Pi$ is \APX-hard, if there is a PTAS-reduction from every problem $\Pi' \in \APX{}$ to $\Pi$.
In other words, if a problem is \APX-hard, it does not admit a polynomial time approximation scheme (PTAS), under standard complexity-theoretic assumptions.
Analogously, one can define the classes $f(n)$-\APX, which consist of the \NP{} optimization problems that admit a polynomial time $f(n)$ approximation.
Of those classes, the class \logAPX{} is of particular interest for us.

To show our hardness results, we use two types of reductions: L-reductions and strict reductions.
Both (in our case) are a reduction from a minimization problem $\Pi_1$ to minimization problem $\Pi_2$ given by a pair of functions $f, g$ with the following properties:
\begin{itemize}
	\item The functions $f$ and $g$ are computable in polynomial time.
	\item If $I_1$ is an instance of $\Pi_1$, then $f(I_1)$ is an instance of $\Pi_2$.
	\item If $S_2$ is a solution to an instance $f(I_1)$ of $\Pi_2$, then $g(I_1, S_2)$ is a solution to $I_1$.
\end{itemize}
We use $\OPT_{\Pi}(\cdot)$ to refer to the cost of the optimal solution for $\Pi$ of a given instance, and $\text{cost}_{\Pi}(\cdot)$ to refer to the cost for $\Pi$ of a given solution.
An L-reduction additionally fulfills the following two properties:
\begin{itemize}
	\item There is a positive constant $\alpha \in \mathbb R$ satisfying $\OPT_2(f(I_1)) \leq \alpha \cdot \OPT_1(I_1)$ for all instances $I_1$ of $\Pi_1$.
	\item There is a positive constant $\beta \in \mathbb R$ satisfying $|\text{cost}_1(g(I_1, S_2)) - \OPT_1(I_1)| \leq \beta \cdot |\text{cost}_2(S_2) - \OPT_2(f(I_1))|$ for all instances $I_1$ of $\Pi_1$ and for all solutions $S_2$ of $f(I_1)$.
\end{itemize}
If $\Pi_2$ admits an approximation ratio of $1 + c$, then $\Pi_1$ admits an approximation ratio of $1 + \alpha\beta 
c$, by transformation above inequalities.
On the other hand, a strict reduction additionally fulfills the following property:
\begin{itemize}
	\item For every instance $I_1$ of $\Pi_1$ and for all solutions $S_2$ of $f(I_1)$, \[\frac{\text{cost}_1(g(I_1, S_2))}{\OPT_1(I_1)} \leq \frac{\text{cost}_2(S_2)}{\OPT_2(f(I_1))}\]
\end{itemize}
Additionally, in case of minimization problems, both L-reductions and strict reductions imply the existence of a PTAS-reduction, making them our tools of choice to show \APX-hardness.
For a more in-depth introduction to approximation preserving reductions, we refer the reader to a survey by Crescenzi \cite{crescenzi1997short}.

\section{APX-hardness}
\label{section:apx}

We analyze the \APX-hardness of \cov{}.
For $\delta \geq \frac{3}{2}$, \cov{} turns out to be \logAPX-hard, which is the same threshold where the parameterized complexity of \cov{} jumps from FPT to W[2]-hard \cite{DBLP:journals/mp/HartmannLW22}.
On the other hand, for every $\delta < \frac{3}{2}$ that is not a unit fraction, \cov{} is \APX-hard.
When $\delta$ is a unit fraction, the problem becomes solvable optimally in polynomial time \cite{DBLP:journals/mp/HartmannLW22}, and therefore cannot be \APX-hard under standard complexity-theoretic assumptions.


\subsection{Asymptotically Optimal Approximation for $\delta \geq \frac{3}{2}$}

When designing an approximation algorithm for \cov{}, one difficulty to overcome is that $P(G)$ has infinite size.
Our main tool to overcome this is the following lemma.
We call a rational number \emph{$b$-simple}, if it can be written as $\frac{a}{b}$ for two natural numbers $a,b$.
A $\delta$-cover $S \subseteq P(G)$ is called $b$-simple, if for every $p \in S$, the distance of $p$ to the nearest vertex is $b$-simple.

\begin{lemma}[\cite{DBLP:journals/mp/HartmannLW22}]
	\label{lemma:cov2bsimple}
	Let $G$ be a graph and let $\delta = \frac{a}{b}$ with $a,b \in \mathbb{N}^+$.
	Then there exists an optimal $\delta$-cover $S^* \subseteq P(G)$ that is $2b$-simple.
\end{lemma}

Hence, it suffices to $2b$-simple $\delta$-covers.
Further, as we show, to verify a solution,
	it suffices the consider the $4b$-simple points.
Thus, we may construct an equivalent \textsc{Set Cover} instance,
	which then allows a logarithmic approximation~\cite{johnson1973approximation,lovasz1975ratio,chvatal1979greedy}.

\begin{restatable}{lemma}{covLogApprox}
	\label{lemma:covLogApprox}
	There is an $\mathcal O(\log n)$-approximation algorithm for \cov.
\end{restatable}
\begin{proof}
Because of \cref{lemma:cov2bsimple},
	it suffices to approximate $2b$-simple $\delta$-covers.
Further, it suffices to consider a finite subset of $P(G)$ when dealing with $2b$-simple solutions.
Consider the set $P'(G) \subseteq P(G)$ constructed as follows:
For every edge $\{u, v\} \in E(G)$, the points $p_x = p(u, v, \frac{x}{4b})$ for $x \in \{0, \dots, 4b\}$ are in 
$P'(G)$.
Then a set $S \subseteq P(G)$ that is $2b$-simple $\delta$-covers $P(G)$,
	if and only if it $\delta$-covers $P'(G)$.
The if direction is trivial since $P'(G) \subset P(G)$, so for a contradiction assume $S$ is $2b$-simple and a 
$\delta$-cover for $P'(G)$, but there is a point $p \in P(G)$ that is not covered by $S$.
Then $p = p(u, v, \lambda)$ for some edge $\{u, v\}$, and there is a natural number $x \in \{0, \dots, 4b\}$ such 
that $x \leq \lambda \leq x+1$.
W.l.o.g. assume that $x$ is divisible by $2$, thus $x$ is $2b$-simple.
Therefore $p' = p(u, v, x+1)$ is not covered by $S$, since no point in $S$ can cover $p'$ without covering $p$, due 
to $S$ being $2b$-simple, which is a contradiction to the assumption that $S$ is a $\delta$-cover for $P'(G)$.

Now we can easily construct a (finite) \textsc{Set Cover} instance $(U, \mathcal S)$, taking $U = P'(G)$ as its 
universe, and for every point $p_x = p(u, v, \frac{x}{4b})$ for $x \in \{0, \dots, 4b\}$ where $x$ is even, the set 
$S_{p_x} = \{p \in P'(G) ~|~ \dist{p, p_x} \leq \delta\}$ is added to $\mathcal S$.
Then $U$ has size $4b \cdot |E(G)|$, and the simple greedy approximation algorithm for \textsc{Set Cover} achieves 
an approximation ratio of $\log(4b \cdot |E(G)|) \in \mathcal O (\log n)$ 
\cite{johnson1973approximation,lovasz1975ratio,chvatal1979greedy}.
\end{proof}

We note that if $b$ is not considered to be a constant but part of the input, a result by Tamir~\cite{Tamir87} can be used to obtain a $\delta^* = \frac{a^*}{b^*} \geq \delta$ such that the size of the optimal solution stays the same, and $b^*$ is polynomially bounded in the input size.

\subsection{\logAPX-hardness for $\delta \geq \frac{3}{2}$}

The constructions for \Cref{lemma:logapxhard1,lemma:logapxhard2} were already used in \cite{DBLP:journals/mp/HartmannLW22} to show W[2]-hardness (and thus NP-hardness) of \cov{} for $\delta \geq 2$ by reducing from \textsc{Dominating Set}.
We use them to show \logAPX-hardness, by additionally defining a function that translates \cov{} solutions to \textsc{Dominating Set} solutions, instead of translating the solution size.
Thus we obtain strict reductions showing \logAPX-hardness, since \textsc{Dominating Set} is \logAPX-hard \cite{chlebik2008approximation}.

\begin{restatable}{lemma}{logapxhardA}
	\label{lemma:logapxhard1}
	For every $\ell \in \mathbb{N}$ with $\ell \geq 2$ and every $\delta$ with $\ell \leq \delta < \ell + \frac{1}{2}$, \cov{} is \logAPX-hard.
\end{restatable}
\begin{proof}
	We define the following two functions $f$ and $g$, where $f$ maps instances of \textsc{Dominating Set} to \cov{} and $g$ maps solutions of \cov{} back to solutions of the original \textsc{Dominating Set} instances.
	For a graph $G$, $f(G) = H$ is the graph were a path of $\ell-1$ edges is attached to each vertex of $G$.
	For a vertex $v \in V(G)$, we denote the last vertex in the path attached to it in $H$ with $v'$.
	For a $\delta$-cover $S \subseteq P(H)$, we set $g(G, S) = D$ such that for each $p \in S$, $D$ contains the vertex closest to $p$ in $G$, breaking ties arbitrarily.
	
	Let $D \subseteq V(G)$ be a dominating set in $G$, then $D$ is also a $\delta$-cover in $H$, as the point on $w'$ for every neighbor $w$ of a vertex $v \in D$ is $\delta$-covered by the point on $v$. 
	Since $\delta \geq 2$, every point in $P(G)$ is also covered.
	Further, if $S \subseteq P(H)$ is a $\delta$-cover in $H$, then $g(G, S)$ is a dominating set in $G$.
	Assume this is not the case, then there is a vertex $v \in V(G)$, such that $N[v] \cap g(G, S) = \emptyset$.
	Thus there is no point in $S$ on any edge adjacent to $v$ or on the path from $v$ to $v'$.
	Additionally, any point on an edge adjacent to a vertex $w \in N[v]$ must have distance at least $\frac{1}{2}$ from $w$.
	Then the closest point in $S$ to $v'$ has distance at least $\ell + \frac{1}{2} > \delta$, which is a contradiction to $S$ being a $\delta$-cover.
	
	From the fact that the set $g(G, S)$ cannot be larger than the set $S$ by construction of $g(G, S)$, it follows that the optimal solutions of $G$ and $H$ must have the same size, and thus also that
	\begin{align*}
		\frac{|g(G, S)|}{|\OPT_{\textsc{Dominating Set}}(G)|} \leq \frac{|S|}{|\OPT_{\cov}(H)|}
	\end{align*}
	Thus $(f, g)$ is a strict reduction.
\end{proof}

\begin{restatable}{lemma}{logapxhardB}
	\label{lemma:logapxhard2}
	For every $\ell \in \mathbb{N}$ with $\ell \geq 2$ and every $\delta$ with $\ell + \frac{1}{2} \leq \delta < \ell + 1$, \cov{} is \logAPX-hard.
\end{restatable}
\begin{proof}
	This proof is very similar to \Cref{lemma:logapxhard1}.
	For the function $f$, instead of attaching a path of $\ell-1$ edges to every vertex of a graph $G$, we attach a path of length $\ell-2$ edges, and then attach a triangle to the end of each path.
	The function $g$ is defined as in the proof of \Cref{lemma:logapxhard1}.
	By analogous arguments it then follows that $(f, g)$ is a strict reduction.
\end{proof}

The previous two lemmas settle all cases where $\delta \geq 2$.
It thus remains to show \logAPX-hardness for $\frac{3}{2} \leq \delta < 2$.
Since the construction from \Cref{lemma:logapxhard2} breaks down for $\ell = 1$, we use a different construction.
For that, we introduce the \emph{wreath product} of two graphs $G_1$ and $G_2$.
The graph $G_1 \wreath G_2$ is given by $V(G_1 \wreath G_2) = V(G_1) \times V(G_2)$ and 
\[
	E(G_1 \wreath G_2) = \{\{(v_1, v_2), (v_1', v_2')\} ~|~ \{v_1, v_1'\} \in E(G_1), \text{ or } v_1=v_1' \text{ and } \{v_2, v_2'\} \in E(G_2)\}.
\]

\begin{restatable}{lemma}{logapxhardC}
	\label{lemma:logapxhard3}
	For every $\delta$ with $\frac{3}{2} \leq \delta < 2$, \cov{} is \logAPX-hard.
\end{restatable}
\begin{proof}
	Again, we give a strict reduction from \textsc{Dominating Set}.
	For a graph $G$ we define $f(G) = H = G \wreath K_2$, where $K_2$ is the clique on two vertices.
	Let $S \subseteq P(H)$ be a $\delta$-cover of $H$.
	For each $p \in S$, let $(v_1, v_2)$ be the closest vertex to $p$ in $H$ where $v_1 \in V(G)$, breaking ties arbitrarily. 
	Then $v_1 \in D$ and $g(G, S) = D$.
	
	Let $D \subseteq V(G)$ be a dominating set in $G$, then $D$ is also a $\delta$-cover in $H$.
	First observe that $D$ is actually also a dominating set in $H$, by construction of $H$.
	Thus the only interesting case are edges between vertices not in $D$.
	Let $e$ be any such edge.
	Then both endpoints of $e$ have a neighbor in $D$, and since $\delta \geq \frac{3}{2}$, $e$ is completely covered by points in $D$.
	
	Further, if $S \subseteq P(H)$ is a $\delta$-cover in $H$, then $g(G, S) = D$ is a dominating set in $G$.
	Assume this is not the case, then there is a vertex $v \in V(G)$, such that $N[v] \cap g(G, S) = \emptyset$.
	Then $S \cap (N[v] \wreath K_2) = \emptyset$, and also any point in $S$ must have distance at least $\frac{1}{2}$ from all points in $N[v]$.
	Then the point on the center of the edge $\{(v, w_1), (v, w_2)\}$ has distance at least $2 > \delta$ from all points in $S$, which is a contradiction to $S$ being a $\delta$-cover.
	
	From the fact that the set $g(G, S)$ cannot be larger than the set $S$ by construction of $g(G, S)$, it follows that the optimal solutions of $G$ and $H$ must have the same size, and thus also that
	\begin{align*}
		\frac{|g(G, S)|}{|\OPT_{\textsc{Dominating Set}}(G)|} \leq \frac{|S|}{|\OPT_{\cov}(H)|}
	\end{align*}
	Thus $(f, g)$ is a strict reduction.
\end{proof}

An example for the construction from \Cref{lemma:logapxhard3} can be seen in \Cref{fig:reduction-ds-cov}.
\begin{figure}[t]
	\centering
	\begin{tikzpicture}[scale=0.5]
		\node[shape=circle, draw=black] (v1) at (0,0) {};
		\node[shape=circle, draw=black] (v2) at (3,0) {};
		\node[shape=circle, draw=black] (v3) at (3,3) {};
		\node[shape=circle, draw=black] (v4) at (0,3) {};
		
		\draw (v1) -- (v2) -- (v3) -- (v4) -- (v1);
		\node[] (G) at (1.5, -1) {$G$};
		
		\node[shape=circle, draw=black] (v11) at (6.5,-0.5) {};
		\node[shape=circle, draw=black] (v21) at (10.5,-0.5) {};
		\node[shape=circle, draw=black] (v31) at (10.5,3.5) {};
		\node[shape=circle, draw=black] (v41) at (6.5,3.5) {};
		\node[shape=circle, draw=black] (v12) at (7.5,0.5) {};
		\node[shape=circle, draw=black] (v22) at (9.5,0.5) {};
		\node[shape=circle, draw=black] (v32) at (9.5,2.5) {};
		\node[shape=circle, draw=black] (v42) at (7.5,2.5) {};
		
		\foreach \x in {1,2,3,4} {
			\draw (v\x1) -- (v\x2);
		}
		
		\foreach \x in {1,2} {
			\foreach \y in {1,2} {
				\draw (v1\x) -- (v2\y);
				\draw (v2\x) -- (v3\y);
				\draw (v3\x) -- (v4\y);
				\draw (v4\x) -- (v1\y);
			}
		}
		\node[] (G) at (8.5, -1) {$f(G)$};
	\end{tikzpicture}
	\caption{Example for the reduction from \textsc{Dominating Set} to \cov{} for $\delta \in [\frac{3}{2}, 2)$.}
	\label{fig:reduction-ds-cov}
\end{figure}
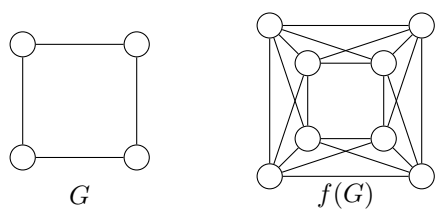

\subsection{\APX-hardness for $\delta < \frac{3}{2}$}
\label{section:apx:delta:leq:32}

To show \APX-hardness for $\delta < \frac{3}{2}$, we first show that the following two self-reductions for \cov{} are L-reductions for graphs that have a large enough solution size.
For a graph $G$ and a natural number $x$ we denote with $G_x$ the $x$-subdivision of $G$,
	that is the graph $G$ where every edge $\{u,v\}$ is replaced by a $u,v$-path of length $x$.

\begin{definition}[Subdivision Reduction]
	\label{def:SubRed}
	This reduction takes $x \in \mathbb{N}^+$ as a parameter. 
	We define $f^S_x(G) = G_x$.
	For an $(x\delta)$-cover of $G_x$,
	we define $g^S_x(G, S') = S$ as the set containing
	\begin{itemize}
		\item every point in $S' \cap V(G)$, and
		\item for every point $p' \in S' \setminus V(G)$,
		the point $p = p(u, v, \frac{\lambda}{x}) \in S$
		where
		$u,v$ are the two closest vertices in $V(G_x) \cap V(G)$
		and $\lambda$ is the distance from $u$ to $p'$ in $G_x$.
	\end{itemize}
\end{definition}

Graph $G$ has an optimal $\delta$-cover of size $k$, if and only if $G_x$ has an optimal $(x \delta)$-cover of size $k$~\cite{DBLP:journals/mp/HartmannLW22}.
The mapping $h$ of a point $p \in P'(G)$
	to the point $p(u,v,\lambda) \in P(G)$
	where $u,v,$ are the two closest vertices in $V(G_x) \cap V(G)$
	and $\lambda$ is the distance from $u$ to $p'$ in $G_x$,
	is a bijection.
We have that a point $p\in P(G_x)$ is $\delta$-covered by a point $q\in P(G_x)$,
	if and only if the point $h(p) \in P(G)$ is $\delta$-covered by the point $h(q) \in P(G)$.
We conclude:
\begin{restatable}{lemma}{subdivisionLRed}
	\label{a:lemma:subdivisionLRed}
	The Subdivision Reduction is an L-reduction for $\alpha = \beta = 1$.
\end{restatable}

Another self-reduction transforms $\delta$ such that
	a $\delta$-cover needs exactly one more point on every edge.

\begin{definition}[Translation Reduction]
	\label{def:TransRed}
	We define $f^T(G) = G$.
	For a $\frac{\delta}{2\delta+1}$-cover $S'$ of $G$,
	we define $g^T(G,S')= S \subseteq P(G)$ as follows:
	For every edge $\{u,v\} \in E(G)$,
	\begin{itemize}
		\item if $|S' \cap P(G[\{u,v\}])|\leq 1$,
		then $S$ contains no point from $e$; and
		\item
		if $S' \cap P(G[\{u,v\}])= \{ p'_1, \dots, p'_k \}$ for a $k\geq 2$
		then $S$ contains the $k-1$ points $p_1, \dots, p_{k-1}$
		where $p_i = p(u, v, \lambda_i)$
		and $\lambda_1 = \mu \cdot (2\delta + 1)$ and $\lambda_{i+1} = \lambda_i + 2\delta$ for $i \in \{2,\dots,k-1\}$
		and $\mu = \min_{i \in [k]} \dist{p_i,u}$.
	\end{itemize}
\end{definition}

\begin{theorem}[\cite{DBLP:journals/mp/HartmannLW22}]
$f^T$ is a polynomial time reduction from \cov{} to \cov[\frac{\delta}{2\delta+1}]
	such that
	$g^T(G,S')$ is a $\delta$-cover of $G$
	for every $\frac{\delta}{2\delta+1}$-cover $S'$ of $G$.
\end{theorem}

Intuitively, a $\frac{\delta}{2\delta+1}$-cover compared to a $\delta$-cover
	must contain exactly one more point on every edge.
Indeed, the minimum $\delta$-cover of $G$ has size $k$,
	if and only if the minimum $\frac{\delta}{2\delta + 1}$-cover of $G$ has 
size $k + |E(G)|$,
	as shown by Hartmann et al.~\cite{DBLP:journals/mp/HartmannLW22}.
Thus the function $g^T$ exactly translates the absolute error of any approximate solution for the $\frac{\delta}{2\delta + 1}$-cover of $G$.

\begin{restatable}{lemma}{translationRed}
	\label{lemma:translationLRed}
	The Translation Reduction $(f^T,g^T)$ is an L-reduction for graphs where $\covn{G} = 
	c\cdot|E(G)|$ for some 
	constant $c > 0$ with $\alpha = {1}/{c}$ and $\beta = 1$.
\end{restatable}

To show \APX-hardness for \cov{} for any interval of $\delta$ with the above reductions, we need an initial interval where \cov{} is \APX-hard.
Further, to use the Translation Reduction,
	we also need \APX-hardness on graphs that have an optimal solution size in $\Omega(|E(G)|)$.
The reductions to show \logAPX-hardness for the interval $\left[\frac{3}{2}, \infty\right)$ are done from \textsc{Dominating Set}.
While \logAPX-hardness already implies \APX-hardness, these reductions do not guarantee that the resulting instance of \cov{} has an optimal solution size in $\Omega(|E(G)|)$.
However, if we use the same reductions from \textsc{Dominating Set} on 3-regular graphs,
	which is \APX-hard \cite{ALIMONTI2000123},
	we obtain \cov-instances that have an optimal solution size in $\Omega(|E(G)|)$.
This is because on 3-regular graphs, any vertex can dominate at most four vertices (including itself).
Therefore, the size of any optimal dominating set is at least $\frac{1}{4}|V(G)|$, and $\frac{3}{2}|V(G)| = |E(G)|$ for 3-regular graphs.

\subsubsection{\APX-hardness for $1 < \delta < \frac{3}{2}$}
\label{sec:apxhardgreater1}
We use the same construction as used in \cite{DBLP:journals/mp/HartmannLW22} to show NP-hardness.
By starting the reduction chain with \textsc{Dominating Set} on 3-regular graphs, we obtain \APX-hardness for the entire interval $(1, \infty)$.
For $i \in \mathbb{N}^+$, define $\alpha_i = \frac{3^i}{3^i-1}$ and $A_i = [\alpha_{i+1}, \alpha_i)$.
Further, let $A_0 = [\frac{3}{2}, \infty)$.
Note that 
\begin{align*}
	3 \cdot \frac{\alpha_i}{2\alpha_i+1} = 3 \cdot \frac{\frac{3^i}{3^i-1}}{2 \cdot \frac{3^i}{3^i-1}+1} = \frac{\frac{3 \cdot 3^i}{3^i-1}}{\frac{2 \cdot 3^i + 3^i - 1}{3^i-1}} = \frac{3^{i+1}}{3^{i+1}-1} &= \alpha_{i+1} \qquad \text{and} \\
	3 \lim_{\delta \rightarrow \infty} \frac{\delta}{2\delta + 1} &= \frac{3}{2}.
\end{align*}
Thus by applying the Translation Reduction and the Subdivision Reduction for $x=3$, we can 
show \cov{} to be 
\APX-hard for $\delta \in A_{i+1}$, by a reduction from \cov{} for $\delta \in A_i$.
Since \cov{} is \APX-hard for $\delta \in A_0$ for graphs with a solution size in $\Omega(|E(G)|)$ due to \Cref{lemma:logapxhard1,lemma:logapxhard2,lemma:logapxhard3}, and $\alpha_i$ converges to $1$ as $i$ increases, we obtain \APX-hardness for the entire interval $(1, \infty)$.

\subsubsection{\APX-hardness for $\delta < 1$}

This section derives \APX-hardness of $\delta$-covering for every $\delta<1$
	that is not a unit-fraction.
Indeed for a unit-fraction $\frac{1}{b}$ for an integer $b$,
	a minimum $\frac{1}{b}$-covering set can be computed in polynomial time,
	as stated in \Cref{lemma:covering:p}.

\begin{restatable}{theorem}{apxhardevenb}
	\label{lemma:apxhardevenb}
	\cov{} is \APX-hard for every $\delta$ that is not a unit-fraction.
\end{restatable}
\begin{proof}
We first obtain \APX-hardness for $\delta$ in the interval $\left(\frac{1}{3}, \frac{1}{2}\right)$,
	then follow \APX-hardness for $\delta$ in an interval $(\frac{1}{b+1},\frac{1}{b})$
	with an \emph{even} positive integer $b$,
	and finally obtain \APX-hardness also for $\delta$ in an interval $(\frac{1}{b+1},\frac{1}{b})$
	for every positive integer $b$.

From \Cref{sec:apxhardgreater1},
	we already know \APX-hardness for $\delta$ in the interval $(1, \infty)$.
By applying the Translation Reduction we also obtain \APX-hardness for $\delta$
	the interval $\left(\frac{1}{3}, \frac{1}{2}\right)$, as $\lim_{\delta \rightarrow 1} \frac{\delta}{2\delta + 1} = \frac{1}{3}$
	and 
	$\lim_{\delta \rightarrow \infty} \frac{\delta}{2\delta + 1} = \frac{1}{2}$.
Further, for $\delta = \frac{1}{b}$ for an integer $b$,
	we note that $\frac{\delta}{2\delta + 1} = \frac{1}{b+2}$.
	Thus repeatedly applying the Translation Reduction yields \APX-hardness for every 
	interval $(\frac{1}{b+1}, 
	\frac{1}{b})$ with even $b$.
Finally, for every  rational number $\delta$, there are $a, b \in \mathbb{N}^+$ such that $b$ is even and ${\frac{a}{b+1} < \delta < \frac{a}{b}}$~\cite{DBLP:journals/mp/HartmannLW22}.
Applying the Subdivision Reduction $a$ times then also yields \APX-hardness for
	the interval $({\frac{a}{b+1},\frac{a}{b}})$,
	hence also for $\delta$.
\end{proof}


\section{Hardness under Unique Game Conjecture}
\label{section:hardness:ugc}

This section derives lower bound on the approximation ratio for \cov,
	assuming the Unique Games Conjecture, see~\cite{khotUGC}.
We reduce from \textsc{Vertex Cover}, which is UG-hard to approximate within any factor better than $2$,
	as shown by Khot et al.~\cite{khot2008vertex}.
In particular, assuming the Unique Games Conjecture,
	it is \NP-hard to distinguish graphs that have a weighted vertex cover of size $\half + \varepsilon_1$ and 
	$1 - \varepsilon_2$, where the total weight of all vertices is $1$.
Further, Khot et al.\ note that their construction can be converted to an unweighted graph
	by the same technique as in~\cite{dinur2005hardness}, which essentially duplicates each vertex to model the 
	weights by the number of copies.
Therefore it is UG-hard to decide
	whether a graph with $n$ vertices has a vertex cover of size $\frac{n}{2} + \varepsilon_1$
	or no vertex cover of size $n - \varepsilon_2$.

We give several sets of reductions showing UG-hardness for different ranges of $\delta$,
	as also shown in \Cref{fig:ug-bounds}.
The core construction, given a vertex cover instance consisting of a graph $G$,
	adds a gadget to every vertex $u \in V(G)$,
	see \Cref{figure:ugc:constructions}.
The constructions dealing with the cases where $\frac{4}{5} \leq \delta < 1$ are shown in \Cref{fig:ug-bounds:high}.
For each gadget, we may assume that a minimum $\delta$-cover contains points from this gadget
	that cover a maximum area of points in the original graph $G$,
	some ball of radius $\ell < 1/2$ around each vertex $u \in V(G)$.
For $\delta \geq 1-\ell$, then a vertex cover of $G$ (as points in $P(G)$) covers all the remaining points $P(G)$.

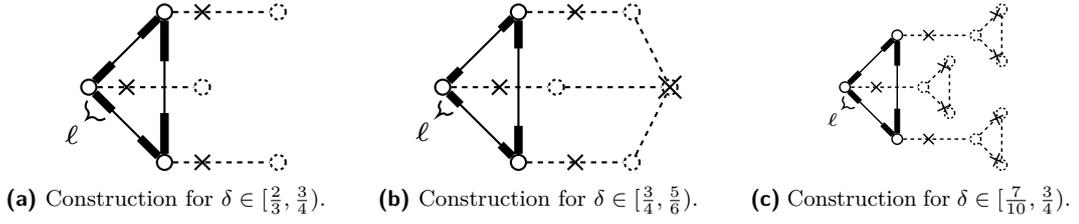
\begin{figure}[t]
	\begin{subfigure}{0.30\textwidth}
		\centering
		\resizebox{!}{0.1\textheight}{
			\begin{tikzpicture}[scale=0.5,
				node/.style = {shape=circle, draw, inner sep=0pt, minimum size=0.25cm},
				textnode/.style = {shape=circle, draw, inner sep=0pt, minimum size=0.4cm},
				smallnode/.style = {shape=circle, draw, inner sep=0pt, minimum size=0.1cm},
				box/.style = {rectangle, fill=gray!20, rounded corners, fill opacity=1, inner sep=1pt},
				cross/.style={cross out, draw=black, minimum size=0.1cm, inner sep=0pt, outer sep=0pt}]
				\node[smallnode] (n1) at (1, 1) {};
				\node[smallnode] (n2) at (0, 0) {};
				\node[smallnode] (n3) at (1, -1) {};
				\draw (n1) -- (n2) -- (n3) -- (n1);
				\node[smallnode, dash pattern=on 0.6pt off 0.6pt] (e1) at (2.5, 1) {};
				\node[smallnode, dash pattern=on 0.6pt off 0.6pt] (e2) at (1.5, 0) {};
				\node[smallnode, dash pattern=on 0.6pt off 0.6pt] (e3) at (2.5, -1) {};
				\draw[dash pattern=on 1pt off 1pt] (n1) -- (e1);
				\draw[dash pattern=on 1pt off 1pt] (n2) -- (e2);
				\draw[dash pattern=on 1pt off 1pt] (n3) -- (e3);
				\node[cross] (p1) at (1.5, 1) {};
				\node[cross] (p2) at (0.5, 0) {};
				\node[cross] (p3) at (1.5, -1) {};
				\node[inner sep=0pt, minimum size=0cm] (n12) at (0.66, 0.66) {};
				\node[inner sep=0pt, minimum size=0cm] (n13) at (1, 0.33) {};
				\draw[ultra thick] (n1) -- (n12) (n1) -- (n13);
				\node[inner sep=0pt, minimum size=0cm] (n21) at (0.33, 0.33) {};
				\node[inner sep=0pt, minimum size=0cm] (n23) at (0.33, -0.33) {};
				\draw[ultra thick] (n2) -- (n21) (n2) -- (n23);
				\node[inner sep=0pt, minimum size=0cm] (n32) at (0.66, -0.66) {};
				\node[inner sep=0pt, minimum size=0cm] (n31) at (1, -0.33) {};
				\draw[ultra thick] (n3) -- (n31) (n3) -- (n32);
				\draw[decoration={brace,mirror,raise=2pt},decorate] (n2) -- node[below left=1pt] {\tiny$\ell$} (n23);
			\end{tikzpicture}
		}
		\caption{Construction for $\delta \in [\frac{2}{3}, \frac{3}{4})$.}
		\label{fig:ugexample1}
	\end{subfigure}
	\hfill
	\begin{subfigure}{0.30\textwidth}
		\centering
		\resizebox{!}{0.1\textheight}{
			\begin{tikzpicture}[scale=0.5,
				node/.style = {shape=circle, draw, inner sep=0pt, minimum size=0.25cm},
				textnode/.style = {shape=circle, draw, inner sep=0pt, minimum size=0.4cm},
				smallnode/.style = {shape=circle, draw, inner sep=0pt, minimum size=0.1cm},
				box/.style = {rectangle, fill=gray!20, rounded corners, fill opacity=1, inner sep=1pt},
				cross/.style={cross out, draw=black, minimum size=0.1cm, inner sep=0pt, outer sep=0pt}]
				\node[smallnode] (n1) at (1, 1) {};
				\node[smallnode] (n2) at (0, 0) {};
				\node[smallnode] (n3) at (1, -1) {};
				\draw (n1) -- (n2) -- (n3) -- (n1);
				\node[smallnode, dash pattern=on 0.6pt off 0.6pt] (e1) at (2.5, 1) {};
				\node[smallnode, dash pattern=on 0.6pt off 0.6pt] (e2) at (1.5, 0) {};
				\node[smallnode, dash pattern=on 0.6pt off 0.6pt] (e3) at (2.5, -1) {};
				\node[smallnode, dash pattern=on 0.6pt off 0.6pt] (e4) at (3, 0) {};
				\draw[dash pattern=on 1pt off 1pt] (n1) -- (e1) -- (e4);
				\draw[dash pattern=on 1pt off 1pt] (n2) -- (e2) -- (e4);
				\draw[dash pattern=on 1pt off 1pt] (n3) -- (e3) -- (e4);
				\node[cross] (p1) at (1.75, 1) {};
				\node[cross] (p2) at (0.75, 0) {};
				\node[cross] (p3) at (1.75, -1) {};
				\node[cross, minimum size=0.15cm] (p4) at (3, 0) {};
				\node[inner sep=0pt, minimum size=0cm] (n12) at (0.75, 0.75) {};
				\node[inner sep=0pt, minimum size=0cm] (n13) at (1, 0.5) {};
				\draw[ultra thick] (n1) -- (n12) (n1) -- (n13);
				\node[inner sep=0pt, minimum size=0cm] (n21) at (0.25, 0.25) {};
				\node[inner sep=0pt, minimum size=0cm] (n23) at (0.25, -0.25) {};
				\draw[ultra thick] (n2) -- (n21) (n2) -- (n23);
				\node[inner sep=0pt, minimum size=0cm] (n32) at (0.75, -0.75) {};
				\node[inner sep=0pt, minimum size=0cm] (n31) at (1, -0.5) {};
				\draw[ultra thick] (n3) -- (n31) (n3) -- (n32);
				\draw[decoration={brace,mirror,raise=2pt},decorate] (n2) -- node[below left=1pt] {\tiny$\ell$} (n23);
			\end{tikzpicture}
		}
		\caption{Construction for $\delta \in [\frac{3}{4}, \frac{5}{6})$.}
		\label{fig:ugexample2}
	\end{subfigure}
	\hfill
	\begin{subfigure}{0.30\textwidth}
		\centering
		\resizebox{!}{0.1\textheight}{
			\begin{tikzpicture}[scale=0.5,
				node/.style = {shape=circle, draw, inner sep=0pt, minimum size=0.25cm},
				textnode/.style = {shape=circle, draw, inner sep=0pt, minimum size=0.4cm},
				smallnode/.style = {shape=circle, draw, inner sep=0pt, minimum size=0.1cm},
				box/.style = {rectangle, fill=gray!20, rounded corners, fill opacity=1, inner sep=1pt},
				cross/.style={cross out, draw=black, minimum size=0.1cm, inner sep=0pt, outer sep=0pt}]
				\node[smallnode] (n1) at (1, 1) {};
				\node[smallnode] (n2) at (0, 0) {};
				\node[smallnode] (n3) at (1, -1) {};
				\draw (n1) -- (n2) -- (n3) -- (n1);
				\node[smallnode, dash pattern=on 0.6pt off 0.6pt] (e1) at (2.5, 1) {};
				\node[smallnode, dash pattern=on 0.6pt off 0.6pt] (e11) at (3, 1.5) {};
				\node[smallnode, dash pattern=on 0.6pt off 0.6pt] (e12) at (3, 0.5) {};
				\node[smallnode, dash pattern=on 0.6pt off 0.6pt] (e2) at (1.5, 0) {};
				\node[smallnode, dash pattern=on 0.6pt off 0.6pt] (e21) at (2, 0.5) {};
				\node[smallnode, dash pattern=on 0.6pt off 0.6pt] (e22) at (2, -0.5) {};
				\node[smallnode, dash pattern=on 0.6pt off 0.6pt] (e3) at (2.5, -1) {};
				\node[smallnode, dash pattern=on 0.6pt off 0.6pt] (e31) at (3, -0.5) {};
				\node[smallnode, dash pattern=on 0.6pt off 0.6pt] (e32) at (3, -1.5) {};
				\draw[dash pattern=on 1pt off 1pt] (n1) -- (e1) -- (e11) -- (e12) -- (e1);
				\draw[dash pattern=on 1pt off 1pt] (n2) -- (e2) -- (e21) -- (e22) -- (e2);
				\draw[dash pattern=on 1pt off 1pt] (n3) -- (e3) -- (e31) -- (e32) -- (e3);
				\node[cross] (p1) at (1.6, 1) {};
				\node[cross, rotate=20] (p11) at (2.9, 1.4) {};
				\node[cross, rotate=340] (p12) at (2.9, 0.6) {};
				\node[cross] (p2) at (0.6, 0) {};
				\node[cross, rotate=20] (p21) at (1.9, 0.4) {};
				\node[cross, rotate=340] (p22) at (1.9, -0.4) {};
				\node[cross] (p3) at (1.6, -1) {};
				\node[cross, rotate=20] (p31) at (2.9, -1.4) {};
				\node[cross, rotate=340] (p32) at (2.9, -0.6) {};
				\node[inner sep=0pt, minimum size=0cm] (n12) at (0.7, 0.7) {};
				\node[inner sep=0pt, minimum size=0cm] (n13) at (1, 0.4) {};
				\draw[ultra thick] (n1) -- (n12) (n1) -- (n13);
				\node[inner sep=0pt, minimum size=0cm] (n21) at (0.3, 0.3) {};
				\node[inner sep=0pt, minimum size=0cm] (n23) at (0.3, -0.3) {};
				\draw[ultra thick] (n2) -- (n21) (n2) -- (n23);
				\node[inner sep=0pt, minimum size=0cm] (n32) at (0.7, -0.7) {};
				\node[inner sep=0pt, minimum size=0cm] (n31) at (1, -0.4) {};
				\draw[ultra thick] (n3) -- (n31) (n3) -- (n32);
				\draw[decoration={brace,mirror,raise=2pt},decorate] (n2) -- node[below left=1pt] {\tiny$\ell$} (n23);
			\end{tikzpicture}
		}
		\caption{Construction for $\delta \in [\frac{7}{10}, \frac{3}{4})$.}
		\label{fig:ugexample3}
	\end{subfigure}
	\caption{Constructions (a), (b) and (c) for \cref{theorem:uglbs} for $x=1$ and input graph $K_3$.
	The dashed edges and vertices are the added gadgets, and the crosses mark the optimal placement of points on 
	them. The 
	thick edge segments are covered by those points. By the choice of $\delta$, for each construction $1 - \delta 
	\leq \ell < \half$. Thus a vertex cover of the original graph covers the remaining edge segments.}
	\label{figure:ugc:constructions}
\end{figure}

\begin{restatable}{theorem}{uglbs}
\label{theorem:uglbs}
	For every $x \in \mathbb N^+$ and $\varepsilon > 0$, it is \UG-hard to approximate \cov{}
	\begin{description}
		\item[(a)] within $1 + \frac{1}{2x+1}$, for any $\delta \in [\frac{x+1}{2x+1}, \frac{2x+1}{4x})$.
		\item[(b)] within $1 + \frac{1}{2x+1+\varepsilon}$, for any $\delta \in [\frac{x+2}{2x+2}, 
		\frac{2x+3}{4x+2})$.
		\item[(c)] within $1 + \frac{1}{2x+5}$, for any $\delta \in [\frac{2x+5}{4x+6}, \frac{x+2}{2x+2})$.
		\item[(d)] within $\frac{6}{5}$, for any $\delta \in [\frac{4}{5}, 1)$.
	\end{description}
\end{restatable}
\begin{proof}[Proof of (a)]
	We give an L-reduction from \textsc{Vertex Cover} on graphs that have a vertex cover of size at least $\half|V(G)|$, i.e., the graphs resulting from the construction of Khot et al.~\cite{khot2008vertex}:
	Given a \textsc{Vertex Cover} instance $G$, let $f(G)$ be the graph where for every vertex $u \in V(G)$ the new vertices 
	$u_1, \dots, u_x$ form a path and $u_1$ is connected to $u$.
	For our definition of $g$, let $\preceq$ be a total ordering of $V(G)$.
	Further, given a $\delta$-cover $S$ of the graph $f(G)$, let 
	\begin{align*}
		g(G, S) &= \{u \in V(G) ~|~ p = p(u, v, \lambda) \in S \text{ 
			with } \lambda < \half \text{ or } \lambda = \half, u \prec v\} \\
		&\phantom{{}={}} \cup \{u \in V(G) ~|~ x < |P(f(G)[\{u,u_1,\dots,u_x\}]) \cap S|\}.
	\end{align*}
	
	By the choice of $\delta$, an optimal solution on $f(G)$ places $x$ points on each path attached to a vertex, 
	w.l.o.g., in distance $2\delta$ of each other.
	Thus, for every vertex $u \in V(G)$,
	the points in the closed ball $\balll{v}{2x\cdot\delta - x}$ are already covered by a 
	point on the path attached to $u$.
	For $\delta \geq \frac{x+1}{2x+1}$, the radius of that ball satisfies $2x\cdot\delta - x \geq 
	2x\cdot\frac{x+1}{2x+1} - x = \frac{2x(x+1) - (2x+1)x}{2x+1} = \frac{x}{2x+1} \geq 1 - \delta$ and for $\delta < 
	\frac{2x+1}{4x}$, it satisfies $2x\cdot\delta - x < 2x\cdot\frac{2x+1}{4x} - x = \half$.
	Thus the midpoint of each edge of $G$ remains uncovered by the points on the added paths, and also a vertex 
	cover of $G$ covers all remaining uncovered points of $f(G)$.
	Since $\OPT_{\textsc{Vertex Cover}}(G) \geq \half|V(G)|$, it follows that $\OPT_{\cov}(f(G)) \leq (2x+1) \cdot \OPT_{\textsc{Vertex Cover}}(G)$.
	
	Further, it is easy to see, that the absolute error of a solution is not increased by $g$:
	Each vertex in $g(G, S)$ corresponds to a point in $S$.
	By the argument above the optimal solution of $f(G)$ has $x \cdot |V(G)| + \text{vc}(G)$ points,
	and there are at least $x \cdot |V(G)|$ points in $S$ that do not get translated to a vertex in $g(G, S)$.
	
	Thus $(f, g)$ is an L-reduction with $\alpha = 2x+1$ and $\beta = 1$.
\end{proof}
\begin{proof}[Proof of (b)]
	Again we give an L-reduction $(f, g)$ from vertex cover on graphs that have a vertex cover of size at least 
	$\half|V(G)|$.
	Let $f$ be defined as in the proof for part (a), except the last vertex $u_x$ of each path is connected to a new vertex 
	$u^*$.
	The function $g$ also needs to be slightly modified to account for $u^*$: 
	\begin{align*}
		g(G, S) &= \{u \in V(G) ~|~ p = p(u, v, \lambda) \in S \text{ with } \lambda < \half \text{ or } \lambda = \half \text{ and } u \prec v\} \\
		&\phantom{{}={}} \cup \{u \in V(G) ~|~ x < |P(f(G)[\{u,u_1,\dots,u_x,u^*\}]) \cap S \setminus \ball{u^*}{\delta}|\}.
	\end{align*}
	Further, if $g(G, S)$ would output $V(G)$, it outputs $V(G) \setminus \{v\}$ for an arbitrary $v \in V(G)$ instead, to make sure the absolute error does not increase when dealing with certain large solutions $S$ for \cov.
	Note that $V(G) \setminus \{v\}$ is a vertex cover for every graph $G$.
	
	By an analogous argument to the one of part (a), $\OPT_{\cov}(f(G)) \leq (2x+1+\varepsilon) \cdot \OPT_{\vc}(G)$ for any 
	$\varepsilon > 0$.
	
	If a solution $S$ of $f(G)$ places a point in $\ball{u^*}{\delta}$, the absolute error of $g(G, S)$ is not larger 
	than that of $S$ by the same arguments as in the proof of part (a).
	\begin{restatable}{claim}{pointOnApexVertex}
		\label{claim:pointOnApexVertex}
		If a solution $S$ of $f(G)$ places no point in $\ball{u^*}{\delta}$, $|g(G, S)| = |V(G)| - 1$.
	\end{restatable}
	\begin{claimproof}
		If there is no point in $\ball{u^*}{\delta}$, $p = p(u^*, u_x, \delta) \in S$ for every $u \in V(G)$.
		W.l.o.g., there are $x-1$ more points with distance $2\delta$ of each other on the path $u, \dots, u_x, u^*$.
		Since the path $u, \dots, u_x, u^*$ has length $x+1$ and $x+1-2\delta x < x + 1 - 2\frac{2x+3}{4x+2}x = 
		\frac{1}{2x+1}$, the point $p = p(u, u_1, \frac{1}{2x+1})$ is not covered by the $x$ points already on the path 
		$u, \dots, u_x, u^*$.
		As $(\balll{p}{\delta} \cap P(G)) \subseteq (\ball{u}{\half} \cap P(G))$, $S$ places at least $x+1$ points in 
		$\ball{u}{\half} \cup P(f(G)[u,\dots,u_x,u^*]) \setminus \ball{u^*}{\delta}$ for every $u \in V(G)$.
		Thus, $g(G, S)$ outputs $V(G) \setminus \{v\}$ for an arbitrary $v \in V(G)$ and the claim follows.
	\end{claimproof}
	Therefore in this case, the absolute error of $g(G, S)$ is not larger than that of $S$, again by the same argument 
	as in the proof of part (a).
	Thus $(f, g)$ is an L-reduction with $\alpha = 2x+1+\varepsilon$ for any $\varepsilon > 0$ and $\beta = 1$.
\end{proof}
\begin{proof}[Proof of (c)]
	We use the same reduction as for the proof of part (a), but to each vertex $u_x$ we attach two additional vertices 
	$u'_x$ and $u''_x$ that form a triangle with $u_x$.
	Then by analogous arguments to those in the proof of part (a), we get an L-reduction with $\alpha = 2x+5$ and $\beta = 
	1$.
\end{proof}

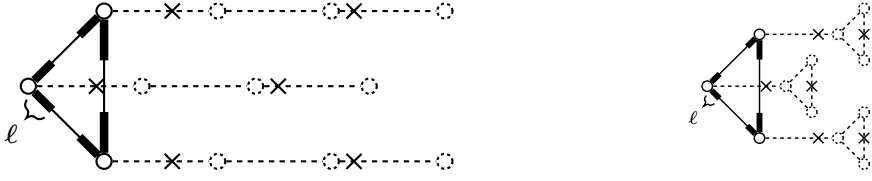
\begin{figure}[t]
	\begin{subfigure}{0.48\textwidth}
		\centering
		\resizebox{!}{0.1\textheight}{
			\begin{tikzpicture}[scale=0.5,
				node/.style = {shape=circle, draw, inner sep=0pt, minimum size=0.25cm},
				textnode/.style = {shape=circle, draw, inner sep=0pt, minimum size=0.4cm},
				smallnode/.style = {shape=circle, draw, inner sep=0pt, minimum size=0.1cm},
				box/.style = {rectangle, fill=gray!20, rounded corners, fill opacity=1, inner sep=1pt},
				cross/.style={cross out, draw=black, minimum size=0.1cm, inner sep=0pt, outer sep=0pt}]
				\node[smallnode] (n1) at (1, 1) {};
				\node[smallnode] (n2) at (0, 0) {};
				\node[smallnode] (n3) at (1, -1) {};
				\draw (n1) -- (n2) -- (n3) -- (n1);
				\node[smallnode, dash pattern=on 0.6pt off 0.6pt] (e1) at (2.5, 1) {};
				\node[smallnode, dash pattern=on 0.6pt off 0.6pt] (e11) at (4, 1) {};
				\node[smallnode, dash pattern=on 0.6pt off 0.6pt] (e12) at (5.5, 1) {};
				\node[smallnode, dash pattern=on 0.6pt off 0.6pt] (e2) at (1.5, 0) {};
				\node[smallnode, dash pattern=on 0.6pt off 0.6pt] (e21) at (3, 0) {};
				\node[smallnode, dash pattern=on 0.6pt off 0.6pt] (e22) at (4.5, 0) {};
				\node[smallnode, dash pattern=on 0.6pt off 0.6pt] (e3) at (2.5, -1) {};
				\node[smallnode, dash pattern=on 0.6pt off 0.6pt] (e31) at (4, -1) {};
				\node[smallnode, dash pattern=on 0.6pt off 0.6pt] (e32) at (5.5, -1) {};
				\draw[dash pattern=on 1pt off 1pt] (n1) -- (e1) -- (e11) -- (e12);
				\draw[dash pattern=on 1pt off 1pt] (n2) -- (e2) -- (e21) -- (e22);
				\draw[dash pattern=on 1pt off 1pt] (n3) -- (e3) -- (e31) -- (e32);
				\node[cross] (p1) at (1.9, 1) {};
				\node[cross] (p11) at (4.3, 1) {};
				\node[cross] (p2) at (0.9, 0) {};
				\node[cross] (p2) at (3.3, 0) {};
				\node[cross] (p3) at (1.9, -1) {};
				\node[cross] (p31) at (4.3, -1) {};
				\node[inner sep=0pt, minimum size=0cm] (n12) at (0.66, 0.66) {};
				\node[inner sep=0pt, minimum size=0cm] (n13) at (1, 0.33) {};
				\draw[ultra thick] (n1) -- (n12) (n1) -- (n13);
				\node[inner sep=0pt, minimum size=0cm] (n21) at (0.33, 0.33) {};
				\node[inner sep=0pt, minimum size=0cm] (n23) at (0.33, -0.33) {};
				\draw[ultra thick] (n2) -- (n21) (n2) -- (n23);
				\node[inner sep=0pt, minimum size=0cm] (n32) at (0.66, -0.66) {};
				\node[inner sep=0pt, minimum size=0cm] (n31) at (1, -0.33) {};
				\draw[ultra thick] (n3) -- (n31) (n3) -- (n32);
				\draw[decoration={brace,mirror,raise=2pt},decorate] (n2) -- node[below left=1pt] {\tiny$\ell$} (n23);
			\end{tikzpicture}
		}
		\caption{Construction for \Cref{theorem:uglbs}(d) for $\delta \in [\frac{4}{5}, \frac{7}{8})$.}
		\label{fig:ugexample4a}
	\end{subfigure}
	\hfill
	\begin{subfigure}{0.48\textwidth}
		\centering
		\resizebox{!}{0.1\textheight}{
			\begin{tikzpicture}[scale=0.5,
				node/.style = {shape=circle, draw, inner sep=0pt, minimum size=0.25cm},
				textnode/.style = {shape=circle, draw, inner sep=0pt, minimum size=0.4cm},
				smallnode/.style = {shape=circle, draw, inner sep=0pt, minimum size=0.1cm},
				box/.style = {rectangle, fill=gray!20, rounded corners, fill opacity=1, inner sep=1pt},
				cross/.style={cross out, draw=black, minimum size=0.1cm, inner sep=0pt, outer sep=0pt}]
				\node[smallnode] (n1) at (1, 1) {};
				\node[smallnode] (n2) at (0, 0) {};
				\node[smallnode] (n3) at (1, -1) {};
				\draw (n1) -- (n2) -- (n3) -- (n1);
				\node[smallnode, dash pattern=on 0.6pt off 0.6pt] (e1) at (2.5, 1) {};
				\node[smallnode, dash pattern=on 0.6pt off 0.6pt] (e11) at (3, 1.5) {};
				\node[smallnode, dash pattern=on 0.6pt off 0.6pt] (e12) at (3, 0.5) {};
				\node[smallnode, dash pattern=on 0.6pt off 0.6pt] (e2) at (1.5, 0) {};
				\node[smallnode, dash pattern=on 0.6pt off 0.6pt] (e21) at (2, 0.5) {};
				\node[smallnode, dash pattern=on 0.6pt off 0.6pt] (e22) at (2, -0.5) {};
				\node[smallnode, dash pattern=on 0.6pt off 0.6pt] (e3) at (2.5, -1) {};
				\node[smallnode, dash pattern=on 0.6pt off 0.6pt] (e31) at (3, -0.5) {};
				\node[smallnode, dash pattern=on 0.6pt off 0.6pt] (e32) at (3, -1.5) {};
				\draw[dash pattern=on 1pt off 1pt] (n1) -- (e1) -- (e11) -- (e12) -- (e1);
				\draw[dash pattern=on 1pt off 1pt] (n2) -- (e2) -- (e21) -- (e22) -- (e2);
				\draw[dash pattern=on 1pt off 1pt] (n3) -- (e3) -- (e31) -- (e32) -- (e3);
				\node[cross] (p1) at (2.125, 1) {};
				\node[cross] (p11) at (3, 1) {};
				\node[cross] (p2) at (1.125, 0) {};
				\node[cross] (p21) at (2, 0) {};
				\node[cross] (p3) at (2.125, -1) {};
				\node[cross] (p31) at (3, -1) {};
				\node[inner sep=0pt, minimum size=0cm] (n12) at (0.75, 0.75) {};
				\node[inner sep=0pt, minimum size=0cm] (n13) at (1, 0.5) {};
				\draw[ultra thick] (n1) -- (n12) (n1) -- (n13);
				\node[inner sep=0pt, minimum size=0cm] (n21) at (0.25, 0.25) {};
				\node[inner sep=0pt, minimum size=0cm] (n23) at (0.25, -0.25) {};
				\draw[ultra thick] (n2) -- (n21) (n2) -- (n23);
				\node[inner sep=0pt, minimum size=0cm] (n32) at (0.75, -0.75) {};
				\node[inner sep=0pt, minimum size=0cm] (n31) at (1, -0.5) {};
				\draw[ultra thick] (n3) -- (n31) (n3) -- (n32);
				\draw[decoration={brace,mirror,raise=2pt},decorate] (n2) -- node[below left=1pt] {\tiny$\ell$} (n23);
			\end{tikzpicture}
		}
		\caption{Construction for \Cref{theorem:uglbs}(d) for $\delta \in [\frac{7}{8}, 1)$.}
		\label{fig:ugexample4b}
	\end{subfigure}
	\caption{The dashed edges and vertices are the added gadgets, and the crosses mark the optimal placement of points on them. The thick edge segments are covered by those points. By the choice of $\delta$, for both constructions $1 - \delta \leq \ell < \half$. Thus a vertex cover on the original graph covers the remaining edge segments, and the shown points don't cover the entire graph on their own.}
	\label{fig:ugexample4}
\end{figure}

\begin{proof}[Proof of (d)]
	For $\delta \in [\frac{4}{5}, \frac{7}{8})$, we reuse the construction in the proof of part (a) with $x=3$.
	Then each of the added paths must contain at least two points.
	With analogous arguments to the proof of part (a), we obtain that this is an L-reduction with $\alpha = 5$ and $\beta = 
	1$.
	
	For $\delta \in [\frac{7}{8}, 1)$, we reuse the construction of the proof of part (c) with $x=1$.
	Then each of the added paths with a triangle must contain at least two points.
	Again with analogous arguments to the proof of part (a), we obtain that this is an L-reduction with $\alpha = 5$ and 
	$\beta = 1$.
\end{proof}

We note that every lower bound for a $\delta \in (\frac{1}{2}, \frac{3}{4})$
	also holds for $2\delta \in (1, \frac{3}{2})$,
	by applying Subdivision Reduction (see \Cref{def:SubRed}) for $x = 2$.


\section{Approximation Algorithm for $1 <\delta < \frac{3}{2}$}
\label{section:approx:large}

This section derives the approximation algorithms for $\delta> 1$.
We do so by constructively bounding the size of minimum $\delta$-cover $S_\delta$
	in the size of a $1$-cover by some factor $\alpha$.
Then a polynomial time computable $1$-cover constitutes an $\alpha$-approximation.
For $\delta < {3}/{2}$,
	we construct a $1$-cover of size at most twice that of a minimum $\delta$-cover,
	while for $\delta < {5}/{4}$ we reduce this factor to $5/3$
	and for $\delta < 7/6$ we further reduce this factor to $3/2$.

\renewcommand{\putmarkx}[1]{
	\pgfmathsetmacro{\xval}{(\x-0.5)*20}
	\draw (\xval,0) -- (\xval,-0.15);
	\node (putmarkx) at (\xval, -0.5) {$#1$};
}
\renewcommand{\putmarky}[1]{
	\pgfmathsetmacro{\yval}{(\y-1)*5}
	\draw (-0.2, \yval) -- (0, \yval);
	\node (putmarky) at (-0.5+0.1*\c, \yval) {$#1$};
}

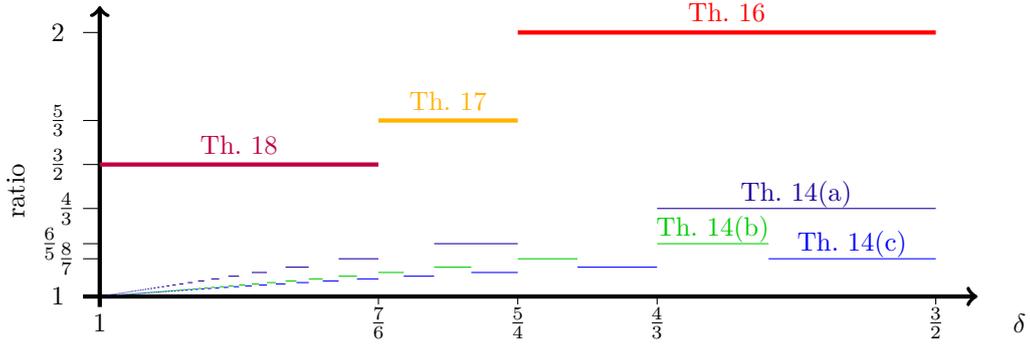
\begin{figure}[t]
	\centering
	\begin{tikzpicture}[yscale=0.7,xscale=1.1]
		\def\colora{dkblue}
		\def\colorb{lgreen}
		\def\colorc{blue}
		\def\colord{cyan}
		\def\colore{purple}
		\def\colorf{lorange}
		\def\colorg{red}
		\draw[->, ultra thick] (-0.2, 0) -- (10.5, 0);
		\node at (11, -0.5) {$\delta$};
		\draw[->, ultra thick] (0, -0.2) -- (0, 5.5);
		\node[rotate=90] at (-1, 2) {ratio};
		
		\def\x{0.5}\putmarkx{1}
		\foreach \a/\b in {7/6,5/4,4/3,3/2}{ 
			\def\x{(\a/\b)-0.5}\putmarkx{\frac{\a}{\b}}
		}
		
		\def\c{0}\def\y{2}\putmarky{2}
		\foreach \a/\b/\c in {8/7/1, 6/5/-1, 4/3/1, 3/2/0, 5/3/0} {
			\def\y{\a/\b}\putmarky{\frac{\a}{\b}}
		}
		\def\y{1}\putmarky{1}
		
		\foreach \x in {1, ..., 100} {
			\pgfmathsetmacro{\from}{((\x+1)/(2*\x+1)-0.5)*40}
			\pgfmathsetmacro{\to}{((2*\x+1)/(4*\x)-0.5)*40}
			\pgfmathsetmacro{\val}{((2*\x+2)/(2*\x+1)-1)*5}
			\ifthenelse{\x=1}{
				\draw[-, color=\colora] (\from, \val) -- node[above, yshift=-3] {Th.~\ref{theorem:uglbs}(a)} (\to, \val);
			}{
				\draw[-, color=\colora] (\from, \val) -- (\to, \val);
			}
			
			\pgfmathsetmacro{\from}{((\x+2)/(2*\x+2)-0.5)*40}
			\pgfmathsetmacro{\to}{((2*\x+3)/(4*\x+2)-0.5)*40}
			\pgfmathsetmacro{\val}{((2*\x+2)/(2*\x+1)-1)*5}
			\ifthenelse{\x=1}{
			}{\ifthenelse{\x=2}{
					\draw[-, color=\colorb] (\from, \val) -- node[above, yshift=-3] {Th.~\ref{theorem:uglbs}(b)} (\to, 
					\val);
				}{
					\draw[-, color=\colorb] (\from, \val) -- (\to, \val);
				}
			}
			
			\pgfmathsetmacro{\from}{((2*\x+5)/(4*\x+6)-0.5)*40}
			\pgfmathsetmacro{\to}{((\x+2)/(2*\x+2)-0.5)*40}
			\pgfmathsetmacro{\val}{((2*\x+6)/(2*\x+5)-1)*5}
			\ifthenelse{\x=1}{
				\draw[-, color=\colorc] (\from, \val) -- node[above, yshift=-3] {Th.~\ref{theorem:uglbs}(c)} (\to, \val);
			}{
				\draw[-, color=\colorc] (\from, \val) -- (\to, \val);
			}
 		}
 		\pgfmathsetmacro{\x}{1}
		\pgfmathsetmacro{\from}{((1/2)-0.5)*20}
		\pgfmathsetmacro{\to}{((4/6)-0.5)*20}
		\draw[-, color=\colore, ultra thick] (\from, 2.5) -- node[above] {Th.~\ref{lemma:approxsmaller76}} (\to, 2.5);

 		\pgfmathsetmacro{\x}{1}
		\pgfmathsetmacro{\from}{((4/6)-0.5)*20}
		\pgfmathsetmacro{\to}{((3/4)-0.5)*20}
		\pgfmathsetmacro{\val}{((5/3)-1)*5}
		\draw[-, color=\colorf, ultra thick] (\from, \val) -- node[above] {Th.~\ref{lemma:approx76to54}} (\to, 
		\val);
		
 		\pgfmathsetmacro{\x}{1}
		\pgfmathsetmacro{\from}{((\x+2)/(2*\x+2)-0.5)*20}
		\pgfmathsetmacro{\to}{((\x+1)/(2*\x)-0.5)*20}
		\pgfmathsetmacro{\val}{((\x+1)/(\x)-1)*5}
		\draw[-, color=\colorg, ultra thick] (\from, \val)
			-- node[above] {Th.~\ref{lemma:approx54to32}} (\to, \val);
	\end{tikzpicture}
	\caption{Upper bounds (as bold lines) and lower bounds under UGC (as thin lines)
	on the approximation ratio of \cov\ plotted for $\delta \in (1,3/2)$.
	The drawn intervals are half-open with the upper end excluded.
	Applying the Subdivision Reduction to \cref{theorem:uglbs} yields the lower bounds.}
	\label{fig:ug-bounds:high}
\end{figure}

For $\delta \geq {9}/{8}$, we show that these approximation ratios
	are best possible with this approach,
	as we give explicit examples where a $1$-cover and a $\delta$-cover differ asymptotically by these 
	factors. 

\begin{restatable}{theorem}{approxFiveQuarterToThreeHalves}
	\label{lemma:approx54to32}
	For every $\delta < \frac{3}{2}$ and graph $G$,
	$\covn[1]{G} \leq 2 \cdot \covn{G}$.
	Further, for $\delta \geq \frac{5}{4}$, this bound is asymptotically tight.
\end{restatable}
\begin{proof}
	First, we show that $\covn[1]{G} \leq 2 \cdot \covn{G}$ for every $\delta < \frac{3}{2}$.
	For that, let $S_\delta$ be an optimal $\delta$-cover of $G$, for any $\delta \in [1, \frac{3}{2})$.
	We construct a $1$-cover $S_1$ by considering each point $p \in S_\delta$ individually:
	\begin{itemize}
		\item If $p \in S_\delta \cap V(G)$, we add $p$ to $S_1$.
		\item If $p$ is on the interior of an edge $\{u, v\}$, we add $u$ and $v$ to $S_1$.
	\end{itemize}
	Then $|S_1| \leq 2|S_\delta|$.
	Further, $S_1$ is a $1$-cover of $G$.
	For that, consider an arbitrary edge $\{u, v\} \in E(G)$.
	If $S_\delta \cap P(G[\{u, v\}]) \neq \emptyset$, then at least one of $u$ or $v$ is in $S_1$, and therefore any point in 
	$P(G[\{u, v\}])$ is $1$-covered by $S_1$.
	On the other hand, if $S_\delta \cap P(G[\{u, v\}]) = \emptyset$, there is an edge $\{v, w\} \in E(G)$ with $S_\delta \cap 
	(P(G[\{v, w\}]) \setminus \{w\}) \neq \emptyset$, or an edge $\{u, w\} \in E(G)$ with $S_\delta \cap (P(G[\{u, w\}]) \setminus 
	\{w\}) \neq \emptyset$, since $2\delta < 3$.
	Therefore either $u$ or $v$ are contained in $S_1$, and thus any point in $P(G[\{u, v\}])$ is $1$-covered by $S_1$.
	
	Since an optimal $1$-cover can be at most as large as $S_1$, we obtain $\covn[1]{G} \leq 2 \cdot \covn{G}$ for 
	$\delta < \frac{3}{2}$.
	
	Second, we give an infinite family of graphs where this bound is asymptotically tight.
	For $2 < k \in \mathbb N$, let $G_k$ be the graph consisting of $k$ triangles, where one vertex of each triangle 
	is connected to a central vertex $v^*$. 
	An example of this construction is shown in \Cref{fig:cov3254lb}.
	For $\delta \geq \frac{5}{4}$, there is a $\delta$-cover of size $k+1$ for $G_k$:
	The central vertex $v^*$ and the center of the edge of each triangle opposite to the vertex connected to the 
	central vertex.
	On the other hand, a $1$-cover of $G_k$ must place $2$ points on each triangle or its connecting edge:
	In the best case, the edge connecting the triangle to $v^*$ is $1$-covered, for example by a point on $v^*$.
	However, a single point cannot $1$-cover a triangle.
	Thus, the approximation ratio on $G_k$ is $\frac{2k}{k+1}$, which converges to $2$ as $k$ goes to infinity.
\end{proof}

\input{tikz/1-cover-lbs}

\begin{restatable}{theorem}{approxSevenSixthToFiveQuarters}
	\label{lemma:approx76to54}
	For every graph $G$ and $\delta < \frac{5}{4}$,
	$\covn[1]{G} \leq \frac{5}{3} \cdot \covn{G}$.
	Further, for $\delta \geq \frac{7}{6}$, this bound is asymptotically tight.
\end{restatable}
\begin{proof}
	We show that $\covn[1]{G} \leq \frac{5}{3} \cdot \covn{G}$ for every $\delta < \frac{5}{4}$.
	For that, let $S_\delta$ be an optimal $\delta$-cover of $G$, for any $\delta \in [1, \frac{5}{4})$.
	We construct a $1$-cover $S_1$ in two parts.
	First, we construct $S_1'$, a partial $1$-cover, by considering each point $p \in S_\delta$ individually:
	\begin{itemize}
		\item If $p \in S_\delta \cap V(G)$, then $p \in S_1'$.
		\item If $p = p(u, v, \lambda)$ and $\lambda < \half$, then $u \in S_1'$ and otherwise $v \in S_1'$
	\end{itemize}
	Then $|S_1'| = |S_\delta|$.
	Further, every edge of $G$ is either completely $1$-covered or not at all by $S_1'$, as $S_1'$ only contains points on vertices.
	Let $G'$ be the graph induced by the edges that are not $1$-covered by $S_1'$.
	Then $|V(G')| \leq |S_\delta|$, as for every vertex $u \in V(G')$ there must be at least one point in $S_\delta$ on an edge 
	$\{u,v\}$ adjacent to $u$ with $v \notin V(G')$.
	As we observed in \cref{equation:1:cover:size},
	$\covn[1]{G} \leq \frac{2}{3}|V(G)|$ for any graph $G$, 
	and hence $\covn[1]{G'} \leq \frac{2}{3}|S_\delta|$.
	We compute an optimal $1$-cover $S_1''$ of $G'$ in polynomial time using \cref{lemma:covering:p},
	and set $S_1 = S_1' \cup S_1''$.
	It is easy to see that $S_1$ is a $1$-cover of $G$, and that $|S_1| \leq \frac{5}{3}|S_\delta|$.
	
	Since an optimal $1$-cover can be at most as large as $S_1$, we obtain $\covn[1]{G} \leq \frac{5}{3} \cdot 
	\covn{G}$ for $\delta < \frac{5}{4}$.
	
	Further, we give an infinite family of graphs where this bound is asymptotically tight.
	For $2 < k \in \mathbb N$, let $G_k$ be the graph consisting of $k$ triangles, where each vertex of each 
	triangle is connected to a central vertex $v^*$ by a path of $3$ edges. 
	An example of this construction is shown in \Cref{fig:cov5476lb}.
	For $\delta \geq \frac{7}{6}$, there is a $\delta$-cover of size $3k+1$ for $G_k$:
	The central vertex $v^*$ and for each vertex $v$ in a triangle the point with distance $\delta - \half$ from $v$ 
	on the path to $v^*$.
	On the other hand, a $1$-cover of $G_k$ must place $5$ points for each triangle, as shown in 
	\Cref{fig:cov5476lb}.
	Thus, the approximation ratio on $G_k$ is $\frac{5k}{3k+1}$, which converges to $\frac{5}{3}$ as $k$ goes to infinity.
\end{proof}

For $\delta < \frac{7}{6}$,
	we refine this approach in several ways.
We construct the partial $1$-cover $S_1'$ more sensitive to $S_\delta$.
Then we analyze the remaining edges more carefully,
	and finally obtain a good approximation by considering two different alternatives.


\begin{restatable}{theorem}{approxsmallerSevenSixth}
	\label{lemma:approxsmaller76}
	For every graph $G$ and $\delta < \frac{7}{6}$, $\covn[1]{G} \leq \frac{3}{2} \cdot \covn{G}$.
	Further, for $\delta \geq \frac{9}{8}$, this bound is asymptotically tight.
\end{restatable}
\begin{proof}
Let a point $p(u,v,\lambda)$ be \emph{$u$-close} if $\lambda \leq \frac{1}{6}$
	and \emph{centric} if $\frac{1}{6} < \lambda < \frac{5}{6}$.
Conveniently, let an edge be \emph{$u$-close} and be a \emph{center edge}
	if it contains a $u$-close point or a centric point, respectively.
We say that $\delta$-cover $S$ of $G$ is \emph{nice} if it is neat (i.e., satisfies \ref{def:edge}) and
\begin{itemize}
\item
 \labeltext{$($A1$)$}{def:two:centers}
no two center edges are adjacent.
\end{itemize}
\begin{restatable}{claim}{niceMinimumCover} 
	\label{claim:niceMinimumCover}
	For every graph $G$ and $\delta \in [1, \frac{3}{2})$, a nice minimum $\delta$-cover $S_\delta$ exists.
\end{restatable}
\begin{claimproof}
	Consider a minimum $\delta$-cover $S_\delta$ that is neat,
	which exists due to \cref{lemma:n:1}.
	For every vertex $u \in V(G)$, let $p_1 = p(u, v_1, \lambda_1), \dots, p_k = p(u, v_k, \lambda_k)$ denote the points 
	in $S_\delta \cap P(G[N[v]])$, and, w.l.o.g., let $\lambda_1 \leq \dots \leq \lambda_k$. If $k > 1$, we add $v_2, 
	\dots, v_k \in S_\delta$.
	Then $S_\delta$ remains neat and satisfies \ref{def:two:centers}.
\end{claimproof}

Now consider a minimum $\delta$-cover $S_\delta$ of $G$ that is nice.
	We construct a $1$-cover $S_1$ of $G$ from $S_\delta$ in two steps.
	First, we construct a partial $1$-cover $S_1'$:
	\begin{itemize}
		\item For every vertex $u \in V(G)$, if $S_\delta$ contains a $u$-close point, we add $u$ to $S_1'$. In particular, this includes all points on vertices.
		\item For every center edge $\{u, v\} \in E(G)$ with $N(v) \setminus \{u\} \subseteq S_1'$,
		we add $u$ to $S_1'$, unless $v$ was already added to $S_1'$ by this rule.
	\end{itemize}
	Let $G'$ be the graph induced by the edges $\{u, v\} \in E(G)$ where $u, v \notin S_1'$.
	Note that by construction of $S_1'$, there are no edges in $E(G')$ that are $u$-close for some vertex $u \in 
	V(G')$, and the center edges in $E(G')$ form a matching of $G'$.
	Then each edge $\{u, v\} \in E(G')$ satisfies at least one of the following properties.
	By symmetry, assume that $u$ has shorter or equal distance to a point in $S_\delta$ than $v$.
	\begin{description}
		\item[(a)] The edge $\{u, v\}$ is a center edge.
		\item[(b)]
		Vertices $u$ and $v$ are incident to center edges in $E(G')$.
		\item[(c)]
		Vertex $u$ is incident to a center edge in $E(G')$
		and $v$ is adjacent to a vertex $w \in V(G)$ with a $w$-close point in $S_\delta \cap  \{ p(w,v,\lambda) \mid \lambda \in (0,\frac{1}{6}] \}$.		
		\item[(d)]
		Vertex $u$ is incident to a center edge in $E(G')$
		and $v$ is adjacent to a vertex $w \in V(G)$ with a $w$-close point in $S_\delta \setminus \{ p(w,v,\lambda) \mid \lambda \in (0,\frac{1}{6}] \}$.
	\end{description}
We classify each edge $\{u, v\} \in E(G')$ as type-(a), -(b), -(c) or -(d)
	depending on the first property $\{u, v\}$ satisfies.
Let $f$ be the mapping that assigns type-(c) and type-(d) edges
	their unique incident center edge. 
Since $\delta < \frac{7}{6}$, for every type-(c) edge $\{u, v\}$,
	the center edge $f(\{u,v\})$ contains a point in distance $<\half$
	to the common vertex of $\{u,v\}$ and $f(\{u,v\})$.
Further, for every type-(d) edge $\{u, v\}$,
the center edge $f(\{u,v\})$ contains a point in distance $<\frac{1}{3}$
to the common vertex of $\{u,v\}$ and $f(\{u,v\})$.

If $E(G')$ contains a type-(d) edge $\{u,v\}$,
	we add the common vertex of $\{u,v\}$ and $f(\{u,v\})$ to $S_1'$, w.l.o.g., let this be $u$.
Then we delete $u$ from $G'$ and reclassify the type-(b) edges adjacent to the other endpoint of $f(\{u, v\})$ as type-(e).
After this step, $E(G')$ contains one less type-(d) edge.
We repeat this modification until $G'$ does not contain any type-(d) edge.
Finally, we delete all isolated vertices from $G'$.
	
	Now with the only edge types of $G'$ being (a), (b), (c) and (e), we proceed with the second part of the construction of $S_1$.
	For that, we consider each connected component $C$ of $G'$ separately.
	Let $V_a(C)$ denote the vertices of $C$ that are adjacent to a center edge.
	Further, let $V_{ce}(C) = V(C) \setminus V_a(C)$ denote the vertices of $C$ adjacent to a type-(c) or type-(e) 
	edge but not to a center edge.
	Note that both endpoints of a type-(b) edge are in $V_a(C)$.
For a component $C$ of $G'$, let $C^*$ be the component of $G$
	defined by 
	$$G[V(C) \cup \{v ~|~ v \notin V(C), u \in V(C) \text{ and } \{u,v\} \in E(G) \text{ and } |P(G[\{u,v\}]) \cap S_\delta| \geq 1\}].$$
	Informally, each component gets enlarged by the edges containing a point $\delta$-covering parts of the type-(c) and type-(e) edges.
	For every vertex in $V_{ce}(C)$, for every component $C$ of $G'$, there is a unique point in $S_\delta$, by definition of the type-(c) and type-(e) edges.
	Further, this point is contained in $P(G[C^*])$.
	Then for any two components $C_1, C_2$ of $G'$, the corresponding components $C^*_1$ and $C^*_2$ are disjoint.
	Additionally, for every component $C$, there are $\half|V_a(C)|$ unique points in $S_\delta$, that are different from the previous points, as the center edges form a matching in $G'$.
	Then, for each component $C$ of $G'$, we compute a $1$-cover $S_1^C$ depending on the size of $V_a(C)$:
	\begin{enumerate}
		\item If $|V_a(C)| < 2\cdot |V_{ce}(C)|$, we set $S_1^C = V_a(C)$.
		\item If $|V_a(C)| \geq 2\cdot |V_{ce}(C)|$, 
		we use \Cref{lemma:covering:p} to output an optimal $1$-cover of $C$ as $S_1^C$.
	\end{enumerate}
	Finally, we set $S_1 = S_1' \cup \bigcup_{C \in \CC(G')} S_1^C$
		where $\CC(G')$ are the connected components of $G'$.
	By the arguments above, we have $|S_\delta \cap P(G[C^*])| = \half|V_a(C)| + |V_{ce}(C)|$.
	
	In case 1, we have $|S_1 \cap P(G[C^*])| = |V_a(C)| + |V_{ce}(C)|$.
	Further, 
	\begin{align*}
		\half |V_a(C)| < |V_{ce}(C)| 
		&\Leftrightarrow 2 \cdot (|V_a(C)| + |V_{ce}(C)|) < 3 \cdot (\half|V_a(C)| + |V_{ce}(C)|) \\
		&\Leftrightarrow |S_1(C^*)| < \tfrac{3}{2} |S_\delta(C^*)|.
	\end{align*}
	
	In case 2, we have $|S_1 \cap P(G[C^*])| \leq 
	\frac{2}{3} |V_a(C)| + \frac{5}{3}|V_{ce}(C)|$.
	Further, 
	\begin{align*}
		|V_{ce}(C)| \leq \half |V_a(C)| 
		&\Leftrightarrow 4|V_a(C)| + 10|V_{ce}(C)| \leq \tfrac{9}{2}|V_a(C)| + 9|V_{ce}(C)| \\
		&\Leftrightarrow |S_1(C^*)| \leq \tfrac{3}{2} |S_\delta(C^*)|.
	\end{align*}
	
	Thus in both cases, the number of points $S_1$ places on $C^*$ is at most $\frac{3}{2}$ times the number of points $S_\delta$ places on $C^*$.
	Since the components $C^*$ for all $C \in \CC(G)$ are mutually disjoint and
		$S_1$ places the same number of points as $S_\delta$ in the parts of $G$ that are not part of a component $C^*$, we get an overall ratio of $\frac{3}{2}$.
		
	To show that this ratio is best possible for $\delta \geq \frac{9}{8}$, we can use the same construction as in the proof of \Cref{lemma:approx54to32}, but instead of connecting each triangle to the central vertex by an edge, we connect it by a path of three edges.
\end{proof}

\section{Approximation Algorithms for $\frac{1}{2} < \delta < 1$}
\label{section:simple:approx}

This section derives our approximation algorithms for $\delta \in (\half, 1)$.
We begin with a bound for the whole interval.
Since a minimum $\half$-cover for a connected non-tree graph $G$ simply is $V(G)$,
	the following gives a $2$-approximation of a $\delta$-cover of $G$ for $\delta \in (\half, 1)$.

\begin{restatable}{theorem}{boundDeltaSmallerOne}
	\label{lemma:bound:delta:smaller:1}
	$\covn{G} \geq \frac{1}{2}|V(G)|$ for every graph $G$ and $\delta < 1$.
\end{restatable}
\begin{proof}
Dual to \covn{G} is \dispn{G}, 
	which is the maximum size subset of points $I \subseteq P(G)$
	such that no two points in $I$ have distance less than $\delta$.
	It is easy to see that $\covn{G} \geq \covn[1]{G}$ for $\delta < 1$.
	Tamir \cite{Tamir1991} showed that for every graph $G$ and $\delta, \varepsilon > 0$, it holds that $\covn{G} \geq \dispn[(2\delta+\varepsilon)]{G}$.
	Thus also $\covn{G} \geq \dispn[2]{G}$ for $\delta < 1$.
	Further, Hartmann \cite{hartmann2022facility} observed that $\covn[1]{G} + \dispn[2]{G} = |V(G)|$, for every graph $G$.
	It therefore follows that $\covn{G} \geq \frac{1}{2}|V(G)|$ for every graph $G$ and $\delta < 1$.
\end{proof}

\subsection{General Approximation for $\frac{1}{2} < \delta < \frac{2}{3}$}
\label{section:better:approx:medium:delta}

Our family of approximation algorithms for $\delta \in (\half,\frac{2}{3})$,
	rely on bounding the size of a $\delta$-cover linearly in $|V(G)|$.
Intuitively, such a $\delta$-cover $S$ can be smaller than $V(G)$
	on a long path $Q$ by spacing the points $S$ far apart
	such that eventually $Q$ contains an edge $e$ where $S \cap P(G[e])=\emptyset$.

\smallskip

Consider a nice minimum $\delta$-cover $S$ of a graph $G$.
We say that 
	a point $p(u,v,\lambda) \in P(G)$ with $\lambda \in [0,1-\delta)$ is \emph{$u$-close}.
We call a $\delta$-cover $S$ \emph{humble} if it is neat (i.e., satisfies \ref{def:edge}) and

\begin{itemize}
\item \labeltext{$($N2$)$}{def:cycle}
there is no cycle $C$ in $G$
	where each edge contains a point from $S$ in its interior,
\item
\labeltext{$($N3$)$}{def:close}
for every vertex $u \in V(G)$,
	set $S$ contains at most one $u$-close point.
\end{itemize}

\begin{restatable}{lemma}{niceSmallerOne}
	\label{lemma:niceSmallerOne}
	For $\delta \in (\half,\frac{2}{3})$,
	there is a humble minimum $\delta$-cover $S$ of $G$.
\end{restatable}
\begin{proof}
	Consider a minimum $\delta$-cover $S$.
	We modify $S$ without increasing its size
	such that it eventually satisfies property \ref{def:edge}, \ref{def:cycle} and \ref{def:close}.
	We first address \ref{def:close}.
	For every vertex $u \in V(G)$ where $S$ contains a $u$-close point,
	assign $u$-close point of $S$ with minimum distance to $u$.
	Any point is $u$-close to at most one vertex $u$.
	For every not-assigned point $p(u,v,\lambda)$ of $S$ that is $u$-close,
	replace $u$ in $S$ by the non-$u$-close point $p(u,v,1-\delta)$.
	Then the modified set $S$ satisfies \ref{def:close} and is a $\delta$-cover of $G$ of minimum size.
	
	Regarding \ref{def:cycle},
	if there is a cycle $C$, where every edge contains a point of $S$ in its interior,
	we may replace the points of $S \cap P(G[C])$ in $S$ by the vertices of $C$.
	Applying this modification repeatedly for every cycle
	yields a minimum $\delta$-cover $S$ that satisfies \ref{def:cycle}.
	The new set still satisfies \ref{def:close} as none of the new points is in the interior of an edge.
	
	Regarding \ref{def:edge},
	If there is an edge $\{u,v\}$, where $|S \cap P(G[\{u,v\}])| \geq 2$
	but $S \cap P(G[\{u,v\}]) = \{u,v\}$,
	we replace the points $S \cap P(G[\{u,v\}])$ in $S$ by the two points $u,v$.
	Applying this modification for every edge yields a minimum $\delta$-cover $S$ that satisfies \ref{def:edge},
	The new set still satisfies \ref{def:close} and \ref{def:cycle}
	as none of the new points is the interior of an edge.
	Hence this final set $S$ is a humble minimum $\delta$-cover.
\end{proof}

For an integer $x \geq 2$ and $\delta \in [ \frac{x+1}{2x+1}, \frac{x}{2x-1} )$,
	we give an approximation of a $\delta$-cover
	by bounding its size in $|V(G)|$,
	which is the size of a $\half$-cover for connected non-tree graphs $G$.
E.g., for $x=2$ and $\delta \in [\frac{3}{5}, \frac{2}{3})$, this gives a $\frac{3}{2}$-approximation.

\begin{restatable}{lemma}{boundHalfToTwoThirds}
	\label{lemma:bound:1:2:to:2:3}
	Let $x \geq 2$ be integer and $\delta \in [ \frac{x+1}{2x+1}, \frac{x}{2x-1} )$.
	Then $|V(G)| \leq \frac{x+1}{x} \cdot \covn[\delta]{G}$,
	for connected graphs $G$ with $|E(G)| \geq x$.
	This bound is asymptotically tight.
\end{restatable}
\begin{proof}
	Let $S_{\delta}$ be a humble minimum $\delta$-cover of $G$.
Let $E'$ be the set of edges $e$ where $S$ contains a point in the interior of $e$.
Then $E'$ induces a forest in $G$
	by property \ref{def:cycle}.
Consider a component $C$ of $G[E']$.
We claim that $|V(C)| \leq \frac{x+1}{x} |S_\delta \cap P(G[C])|$.
This then also implies $|V(G)| \leq \frac{x+1}{x} |S_\delta|$
	since every point $p \in S_\delta$ is either a vertex
	or is contained in $P(G[C])$ of a component $C$ of $G[E']$.

By property \ref{def:edge} we have $|S_{\delta} \cap P(G[C])| = |E(C)|$.
Let $\beta(C)$ be the set of edges of $G$
	incident to $C$ and incident to $V(G)\setminus C$.
In case $\beta(C)=\emptyset$, then $G[C]$ is the whole graph $G$.
We have that $|S_{\delta}| = |E(C)| = |E(G)| \geq x$ and $|V(C)|\leq x+1$.
That means $|V(G)|  \leq \frac{x+1}{x} |S_\delta|$.

In case $\beta(C)\neq\emptyset$, we show that $|V(C)|\geq x+1$.
We have $|E(C)|\geq x$,
	such that similarly to before, $V(C)$ has size at most
	$\frac{x+1}{x} |S_\delta \cap P(G[C]) |$.
We fix an edge $\{u_0,u_1\}\in\beta(C)$ where $u_1 \in V(C)$.
Inductively, we construct a path $Q= (u_{0},u_1,\dots,u_{x'},u_{x'+1})$ of some length $x'+1\geq 2$.
In step $i$, for increasing $i\geq 0$,
	we identify a point $p_i$ having distance $1-\lambda_i$ to $u_{i+1}$..
For step $i=0$,
	we use that there is a unique $u_{0}$-close point $p_{0}$,
	which is not in the interior of the edge $\{u_0,u_1\}$.
We define $-\lambda_0$
	as the distance of $p_0$ to $u_{0}$.
Consequently, $p_{0}$ has distance $1-\lambda_0$ to $u_1$.
We have $-\lambda_0 \leq \delta-1 \leq 1-\delta$, as $\delta \leq 1$.
We proceed with step $i=1$.
Step $i\geq 1$ is defined as follows.
\begin{itemize}
\item 
If there is a $u_{i}$-close point $p_i = p(u_i,u_{i+1},\lambda_i)$,
	let $\{u_{i},u_{i+1}\}$ be the edge containing the unique $u_i$-close point $p_i$,
	which must be distinct from the previous edge $\{u_{i-1},u_i\}$.
They define vertex $u_{i+1}$ and edge position $\lambda_{i}$.
Then proceed with step $i+1$.
\item
Otherwise, there is an incident edge $\{u_{i},u_{i+1}\}$
	where $S$ contains a point $p(u_i,u_{i+1},\lambda_i)$ with $\lambda_i \geq 1-\delta$,
	which exists as otherwise vertex $u_i$ is not covered.
They define vertex $u_{i+1}$ and edge position $\lambda_i$.
Again $\{u_{i},u_{i+1}\}$ is different from $\{u_{i-1},u_i\}$.
In this case, our procedure terminates.
Let $Q= (u_{0},u_1,\dots,u_{x'},u_{x'+1})$ be the resulting path.
\end{itemize}
	
	We claim that our procedure terminates and that $u_1,\dots,u_{x'+1}$ are distinct vertices in $V(C)$,
	such that $|V(C)|\geq x'+1$.
	Indeed, edge $\{u_i,u_{i+1}\}$ is distinct from its predecessor edge $\{u_{i-1},u_i\}$ by construction.
	By property \ref{def:cycle}, $G[C]$ is a tree
	and hence all of $u_1,u_2,\dots,u_{x'+1}$ are distinct.
	
	We claim that $\lambda_i \leq i(2\delta - 1)$ for $i\in\{0,\dots,x'\}$.
	Our proof works by induction over $i\in\{0,\dots,x'\}$.
	For $i=0$, we have $\lambda_0 \leq 0$ by definition.
	Consider $i\geq 1$.
	In case that $p_i$ is $u_i$-close, no edge incident to $u_i$ other than $\{u_{i},u_{i+1}\}$
	contains a $u_i$-close point, by property \ref{def:close}.
	The nearest point to $u_i$ on edge incident to $u_{i-1}$
	must be $u_{i-1}$-close and hence has distance $1-\lambda_{i-1}$ to $u_i$;
	which particularly is also true for the special case $i=1$ with $\lambda_0$.
	Hence $\lambda_i \leq \lambda_{i-1} + 2\delta -1$
	as otherwise the points $p(u_{i-1},u_i,\lambda)$
	with $\lambda \in (\lambda_{i-1}+\delta, \lambda_i-\delta+1)\neq\emptyset$ are not covered.
	By the induction hypothesis, $\lambda_{i-1} \leq (i-1) (2\delta-1)$.
	Thus $\lambda_i \leq i(2\delta-1)$.
	
	At termination, we have
	$1-\delta \leq \lambda_{x'} \leq x' (2\delta -1)$
	and hence that $x' \geq \frac{1-\delta}{2\delta-1}$, which increases with decreasing $\delta$.
	As $\delta < \frac{x}{2x-1}$, we conclude that $x'> x-1$ and hence $|V(C)| \geq x'+1 \geq x+1$.
	Finally, $|S_\delta \cap P(G[C])| = |E(C)| = |V(C)|-1 $.
	Hence $|V(C)|
	\leq \frac{x+1}{x} |S_{\delta} \cap P(G[C])|$.
	
	\smallskip
	
	The bound of \cref{lemma:bound:1:2:to:2:3} is asymptotically tight.
	Consider the $(x+1)$-subdivision of the star $K_{1,k}$, for an arbitrary large $k$.
	We have that $|V(G)|= x(k+1)+1$.
	For $\delta < \frac{x}{2x-1}$ there is a $\delta$-cover of size $1+kx$
	consisting of the center vertex and all points in distance $2\delta i$ from the center for $i \in \{1, \dots, k\}$.
\end{proof}

To compute a $\frac{x+1}{x}$-approximation of a $\delta$-cover 
	for instances with $|E(G)|<x$, we may use a brute-force algorithm as $x$ can be considered a constant.

\begin{restatable}{theorem}{ToTwoThirdsTheorem}
	\label{lemma:approx:1:2:to:2:3}
	Let $x \geq 2$ be integer and $\delta \in [ \frac{x+1}{2x+1}, \frac{x}{2x-1} )$.
	Then $\covn[\delta]{G}$ allows a polynomial time $\frac{x+1}{x}$-approximation algorithm.
\end{restatable}

\subsection{$3/2$-Approximation Algorithm for $\delta \in [\frac{2}{3},\frac{3}{4})$}

For $\delta$ in the interval $[\frac{2}{3},1)$,
	\Cref{lemma:approx54to32}
	(and similarly the idea of \Cref{lemma:bound:1:2:to:2:3}) only give a $2$-approximation.
In case $\delta < \frac{3}{4}$, we improve this upper bound to $3/2$.
As the lower bound proof of \Cref{theorem:uglbs}(a) reveals,
	the neighbors of the leaves hide a vertex cover instance.
Our algorithm computes an approximate vertex cover for these vertices.
Then our analysis makes use of a carefully chosen partition of the input graph.

\smallskip

We define levels of the input graph $G$ as follows.
Let $L_0$ be the vertices of degree one of $G$,
	and $L_i = \{ u \in V(G) \mid d(u,u_0)=i, u_0 \in L_0\}$, for $i\in\{1,2\}$,
	possibly $L_1 \cap L_2 \neq \emptyset$.
For $i,j\in \{0,1,2\}$, with $i\neq j$, let $E_{i,j}$ be set of tuples $(u_i,u_j)$
	where $\{u_i,u_j\}\in E(G)$ and $u_i \in L_i$ and $u_j \in L_j$.
For $i \in \{0,1,2\}$, let $E_{i,i}$ be the set of edges $E(G) \cap (L_i \times L_i)$.

Our approximation algorithm proceeds as follows.
\begin{enumerate}
\item
Let $L$ be the set of points $p(u_0,u_1,\frac{2}{3})$ for every leaf-edge $(u_0,u_1) \in E_{0,1}$.
\item
We compute a $2$-approximation $X$ of the vertex cover of $G[E_{1,1}]$.
\item
Let $W = V(G)\setminus (L_0 \cup L_1)$ and output $S = L \cup X \cup W$.
\end{enumerate}

Since a $2$-approximation of a vertex cover
	can be computed in polynomial time,
	the above algorithm runs in polynomial time.
Further, we observe that the output $S$ is a $\frac{2}{3}$-cover,
	and hence also a $\delta$-cover for $\delta \in [\frac{2}{3},\frac{3}{4})$.
Indeed, the points $L \cup W$ $\frac{2}{3}$-cover every point on every edge $\{u_1,u_1'\} \in E(G) \setminus E_{1,1}$
	as well as every point in $\balll{u}{\frac{1}{3}}$ for every $u \in L_1$.
The remaining points are the points on every edge $\{u_1,u_1'\}\in E_{1,1}$
	with distance $\leq \frac{2}{3}$ to $u_1$ and to $u_1'$,
	which are covered by the vertex cover $X$ of $G[E_{1,1}]$. 
	
To bound the approximation ratio, we compare the output $S$ with
	a humble minimum (i.e., satisfying \ref{def:edge}-\ref{def:close})
	$\delta$-cover $S^\star$
	that also satisfies the following properties:
\begin{itemize}
\item \labeltext{$($B1$)$}{def:mid}
$p(u_1,u_1',\half) \notin S^\star$ for every edge $\{u_1,u_1'\} \in E_{1,1}$.
\item \labeltext{$($B2$)$}{def:1:2:far}
Every point $p(u_1,u_2,\lambda)\in S^\star$ with $(u_1,u_2)\in E_{1,2}$ and $\lambda\geq  1-\delta$,
	has $\lambda = 1$.
\item \labeltext{$($B3$)$}{def:vc}
$X^{\star\star} \coloneqq \{ u_1 \in L_1 \mid S^\star \cap P_{u_1} \neq \emptyset \} $
	is a vertex cover of the graph $G[E_{1,1}]$,
	where $P_{u_1} \coloneqq \balll{u_1}{\delta - \half} \cup (\ball{u_1}{\half} \cap P(G[E_{1,1}]))$
	for $u_1 \in L_1$.
\end{itemize}

\begin{restatable}{claim}{niceAOneATwoAThree}
	\label{claim:niceA1A2A3}
	There is a humble minimum $\delta$-cover $S^\star$ that satisfies properties \ref{def:mid}-\ref{def:vc}.
\end{restatable}
\begin{claimproof}
	Let $S^\star$ be a humble minimum $\delta$-cover of $G$.
	We modify $S^\star$ such that it remains minimum,
	satisfies properties \ref{def:mid}, \ref{def:1:2:far} and \ref{def:vc}
	and additionally:
	\begin{itemize}
		\item \labeltext{$($B0$)$}{def:leaf}
		For every vertex $u_1 \in L_1$,
		$S^\star$ contains a point in $\balll{u_1}{1-\delta}$.
	\end{itemize}
	First, we modify $S^\star$ such that its satisfies \ref{def:leaf}.
	For every edge $(u_0,u_1) \in E_{0,1}$
	where $|S^\star \cap P(\{u_0,u_1\})|=\{p\}$,
	it holds that $p \in \balll{u_0}{\delta}$,
	and we replace $p$ in $S^\star$ by the point $p(u_0,u_1,\delta)\in S^\star$.
	Every edge $(u_0,u_1) \in E_{0,1}$,
	where $|S^\star \cap P(\{u_0,u_1\})|\geq 2$,
	already satisfies \ref{def:leaf} since by \ref{def:edge} we have $u_1 \in S^\star$.
	After this modification step, $S^\star$ is a minimum $\delta$-cover that satisfies \ref{def:leaf}.
	Further, $S^\star$ still satisfies \ref{def:edge}, \ref{def:cycle}, \ref{def:close}
	since $\{u_0,u_1\}$ contains at most one point from $S^\star$,
	edge $\{u_0,u_1\}$ is not on a cycle and $p(u_0,u_1,\delta)$ is not $u_1$-close.
	
	Regarding \ref{def:mid},
	for every edge $\{u_1,u_1'\}\in E_{1,1}$ where there is a point $p=p(u_1,u_1',\half) \in S^\star$,
	we replace $p$ in $S^\star$ by the point $p(u_1,u_1',1-\delta)$
	(where $u_1$ is an arbitrary chosen end vertex.)
	We note that 
	there are vertices $u_0,u_0'$ with neighborhood $N(u_0)=\{u_1\}$ and $N(u_0')=\{u_1'\}$.
	By property \ref{def:leaf}, $S^\star$ contains a point $q \in \balll{u_1}{1-\delta}$
	as well as a point $q' \in \balll{u_1'}{1-\delta}$.
	Points $q,q'$ cover every point covered by $p=p(u_1,u_1',\half)$
	besides a subset of the points on edge $\{u_1,u_1'\}$ with distance at most $2-2\delta\geq\frac{2}{3}$ to $u_1$.
	Hence $S^\star$ where $p$ is replaced by $p(u_1,u_1',1-\delta)$, remains a $\delta$-cover.
	Applying this modification exhaustively
	yields a minimum $\delta$-cover $S^\star$ that satisfies~\ref{def:mid},
	that still satisfies \ref{def:edge}, \ref{def:cycle}, \ref{def:close} and \ref{def:leaf}
	as the edges with points from $S^\star$ in their interior is the same
	and no new point is close to a vertex.
	
	Regarding~\ref{def:1:2:far},
	for every edge $\{u_1,u_2\}$ with $(u_1,u_2) \in E_{1,2}$,
	if there is a point $p= p(u_1,u_2,\lambda) \in S^\star$ with $\lambda \geq 1-\delta$,
	we replace $p$ in $S^\star$ by $u_2$.
	By property \ref{def:leaf}, $S^\star$ contains a point $q \in \balll{u_1}{1-\delta}$.
	Since $\lambda \geq 1-\lambda$, every point covered by $p$
	is also covered by $q$ or has a shortest path to $p$ via $u_2$.
	Further, $p$ and $u_2$ have distance at most $1+(1-\delta)\leq 2\delta$.
	Hence $S^\star$ where $p$ is replaced by $u_2$, remains a $\delta$-cover.
	Applying this modification exhaustively
	yields a minimum $\delta$-cover $S^\star$ that satisfies~\ref{def:1:2:far}
	that also satisfies \ref{def:edge}, \ref{def:cycle}, \ref{def:close}, \ref{def:leaf}
	and \ref{def:mid},
	since no new point is located in the interior of an edge.
	
	Regarding \ref{def:vc},
	for every edge $\{u_1,u_1'\}\in E_{1,1}$,
	we have $S^\star \cap \balll{p(u_1,u_1',\half)}{\delta} \neq \emptyset$
	in order to $\delta$-cover $p(u_1,u_1',\half)$.
	Recall that $p(u_1,u_1',\half)\notin S^\star$ by \ref{def:mid}.
	Hence for every edge $\{u_1,u_1'\}$ of $G[E_{1,1}]$, we have $S^\star \cap (P_{u_1} \cup P_{u_2})\neq\emptyset$,
	where $P_{u_1} \coloneqq \balll{u_1}{\delta - \half} \cup (\ball{u_1}{\half} \cap P(E_{1,1}))$.
	In other words, $X^{\star\star} = \{ u_1 \in L_1 \mid S^\star \cap P_{u_1} \neq \emptyset \} $
	is a vertex cover of the graph $G[E_{1,1}]$.
	Note that $P_{u_1}$ for $u_1 \in L_1$ are pairwise disjoint.
\end{claimproof}

Our goal is to identify disjoint sets $P_{V^\star}, P_{\CC_{\geq 2}}, P_L \subseteq P(G)$
	such that $S \subseteq P_{V^\star} \cup P_{\CC_{\geq 2}} \cup P_L$ and
	$|S \cap P|\leq \frac{3}{2} |S^\star \cap P|$ for each set $P \in \{P_{V^\star},P_{\CC_{\geq 2}}, P_L\} $,
	which we do below in (I), (II) and (III).
Then also $|S|\leq \frac{3}{2} |S^\star|$,
	and hence our algorithms outputs a $\frac{3}{2}$-approximation.
Let $V^\star \coloneqq S^\star \cap W$
	and let $P_{V^\star} \coloneqq \ball{V^\star}{1}$.
\smallskip

(I)
We claim that $|S \cap P_{V^\star}|\leq |S^\star \cap P_{V^\star}|$.
Any point $p \in S^\star$ contained in the interior of an edge incident to $V^\star \subseteq S^\star$
	contradicts property \ref{def:edge}.
In other words, $S^\star \cap P_{V^\star} = V^\star$.
Since also $S$ satisfies property \ref{def:edge} and $V^\star \subseteq W \subseteq S$,
	analogously to $S^\star$, we have $S \cap P_{V^\star} = V^\star$.
Then, trivially, $|S \cap P_{V^\star}|\leq |S^\star \cap P_{V^\star}|$.

\smallskip

Let $E_\emptyset^\star$ be the set of edges $\{u,v\}\in E(G)$ where $S^\star \cap P(G[\{u,v\}])=\emptyset$,
	which particularly means that $u,v \notin S^\star$.
Let $G^{\star}$ be the graph resulting from $G$ after removing vertices $V^\star \cup L_0 \cup L_1$ from $V(G)$
	with their incident edges and edges $E_\emptyset^\star$ from $E(G)$.
Let $\CC \subseteq 2^{V(G)}$ be the components of $G^\star$.
Let $P_{\CC_{\geq 2}}$ be the union of $P(G[C])$ over all components $C \in \CC$
	where $C$ contains at least $2$ edges.

\smallskip

(II)
We claim that $|S \cap P_{\CC_{\geq 2}} | \leq \frac{3}{2} |S^\star \cap P_{\CC_{\geq 2}}|$.
For every component $C \in \CC_{\geq2}$,
	set $S$ contains $|V(C)|\leq |E(C)|+1$ vertices while $S^\star$ contains $|E(C)|\geq 2$.
Hence $|S \cap P(G[C])| \leq \frac{3}{2} |S^\star \cap P(G[C])|$.
Since components $C\in \CC$ are disjoint, we conclude that $|S \cap P_{\CC_{\geq 2}} | \leq \frac{3}{2} |S^\star \cap P_{\CC_{\geq 2}}|$.

\smallskip

To show (III), we first consider $\CC$ more closely.
It is easy to see that no component $C \in \CC$ consists solely of single vertex $u$.
Indeed, the point of $S^\star$ closest to the vertex $u$
	is either a neighbor $v \in N(u)$
	or
	a point $p(v,u,\lambda)$ for a neighbor $v\in L_1$ with $\lambda<1-\delta$ by \ref{def:1:2:far},
	which in both cases do not $\delta$-cover $u$.
Let $\CC_1\subseteq\CC$ be the subset of components that consist of exactly 1 edge.
Let $X^\star \coloneqq S^\star \cap P_1$ where $P_1 \coloneqq \bigcup_{u_1 \in L_1} P_{u_1}$,
	and $P_{u_1}$ defined as in \ref{def:vc}.
By definition, $|X^{\star\star}| \leq |X^\star|$. 

\begin{restatable}{claim}{mappingThreeHalvesApprox}
	\label{claim:mapping32Approx}
	There is an injective mapping $f: \CC_1 \to X^\star$
	such that every component $C \in \CC_1$
	consists of an edge $\{u,v\}$ incident to the vertex $f(\{u,v\})$.
\end{restatable}
\begin{claimproof}
	Consider a component $C\in \CC_1$ consisting of a single edge $\{u,v\}$.
	Then $S^\star \cap P(G[\{u,v\}])$ contains exactly one point $p(u,v,\lambda)$
	with, up to symmetry, $\lambda \geq \half$, by property \ref{def:edge}.
	We have that $u,v \notin L_0 \cup L_1$ by the definition of $G^\star$.
	Particularly, there is at least one neighbor $u' \in N(u) \setminus  \{v\}$.
	Then there is a point $p(u',u,\lambda') \in S^\star$
	with $\lambda' \geq \frac{3}{2}-2\delta \in (0,\frac{1}{6}]$,
	as otherwise the points around $p(u',u, \frac{3}{2}-\delta)$ are not $\delta$-covered.
	Hence $\{u',u\} \notin E^\star_\emptyset$.
	Further $C$ as a component of $G^\star$ is not incident to an edge $\{v',v''\}$
	with $(v',v'') \in E_{0,1} \cup E_{1,1}$.
	The only remaining possibility is that $(u',u) \in E_{1,2}$.
	We have $\frac{3}{2}-2\delta \leq \delta-\half$ since $\delta \geq \frac{3}{2}$.
	That means $p(u,u',\lambda') \in S^\star \cap P_{u_1}$.
	We assign $f(C)$ to $p(u,u',\lambda')$.
	Since $\lambda'>0$ and no components of $\CC_1$ share incident edges incident to $L_1$,
	mapping $f$ is injective.
\end{claimproof}

(III)
Let $P_{L} \coloneqq P(G[E_{0,1}]) \cup \bigcup_{u_1 \in L_1} P_{u_1} \cup P_{\CC_1}$
	where $P_{\CC_1}$ is the union of $P(G[\{u,v\}])$ for every edge $\{u,v\}$ forming a component in $\CC_1$.
We claim that $|S \cap P_L| \leq \frac{3}{2}|S^\star \cap P_L|$.
Let $L^\star \coloneqq S^\star \cap \balll{L_0}{\delta}$,
	which has $|L^\star| = |L_0| = |L|$.
Let $X^\star_C \coloneqq S^\star \cap \bigcup_{C\in \CC_1} P(G[C])$.
We note that $|X_C^\star| \leq |\CC_1|$ since $S^\star$ contains at most one vertex of every $P(G[C])$,
	by property \ref{def:edge}.
Since further $f: \CC_1 \to X^\star$ is injective, we have $|X_C^\star| \leq |\CC_1| \leq |X^\star|$.
The approximation ratio restricted to $P_L$ then is
$$
\frac{|S \cap P_L|}{|S^\star \cap P_L|}
	= \frac{|L \cup X \cup \bigcup_{C\in\CC_1} V(C)|}{|L^\star \cup X^\star \cup X_C^\star|}
	\leq \frac{|L|+|X|+2|X^\star|}{|L|+2|X^\star|}.
$$
We recall that $|X^{\star\star}|\leq |X^\star|$.
The set $X$, of size at most $|L_1|\leq|L|$,
	is a $2$-approximation of $X^{\star\star}$,
	and hence $|X|\leq \min\{ |L|, 2|X^\star|\}$.
In case $|X^\star|\leq \half|L|$,
	the above ratio is at most $(2|L|+4|X^\star|) / (|L|+ 2|X^\star|)$, which is at most $\frac{3}{2}$.
Otherwise, in case $|X^\star| > \half|L|$,
	the above ration is  at most $(2|L|+2|X^\star|)/(|L|+2|X^\star|)$,
	which also is at most $\frac{3}{2}$.

\smallskip

Sets $P_{V^\star}, P_{\CC_{\geq 2}}, P_L$ are disjoint
	by their definition.
It remains to show that $S \subseteq P_{V^\star} \cup P_{\CC_{\geq 2}} \cup P_L$.
Every point $p \in S \setminus V(G)$, has form $p(u_0,u_1,\frac{2}{3})$ for some edge $(u_0,u_1) \in E_{0,1}$,
	and hence $p \in P(G[E_{0,1}]) \subseteq P_L$.
Every point $p \in S \cap V(G)$
	is either contained in $V^\star$ or in $V(G^\star)$.
In the former case, $p \in P_{V^\star}$.
In the latter case $p \in P_{\CC_1} \subseteq P_L \cup P_{\CC_{\geq 2}}$.
Hence: 

\begin{restatable}{theorem}{approxThreeHalvesToThreeQuartersTheorem}
	\label{lemma:approx:3:2:to:3:4}
	For every $\delta \in [\frac{2}{3}, \frac{3}{4})$,
	there is a $\frac{3}{2}$-approximation algorithm for $\delta$-\textsc{Covering}.
\end{restatable}

There is a gap to the lower bound under UGC of $\frac{4}{3}$, as shown 
\Cref{theorem:uglbs}(a). 
However, our algorithm as a \emph{uniform} algorithm for $\delta \in [\frac{2}{3}, \frac{3}{4})$,
	is reasonably close to the lower bound in the following sense.
The upper and lower bound only differ by a factor of $\frac{3}{2}/\frac{4}{3} = \frac{9}{8}$.
Hence any improvement of the approximation factor for such a uniform algorithm implies
	a better approximation guarantee for the class of graphs $\GG$
	where the minimum size of a $\delta$-cover can be apart by a factor of asymptotically 
	$\frac{9}{8}$. 
An example is
	a path of length $24x$ for $x\geq 1$.
A minimum $\frac{3}{4}$-cover has size $18(x-1)$
	while a minimum $(\frac{3}{4}-\varepsilon)$-cover has size $16(x-1)+1$ for a small enough $\varepsilon>0$.

\section{Approximation Algorithms for $\delta < \frac{1}{2}$}
\label{section:approx:small}

By the Subdivision Reduction (\Cref{def:SubRed}), an upper bound for $\delta > \half$
	implies the same upper bound for $\frac{\delta}{c}$ for $c \in \mathbb{N}^+$.
This can be further improved by using the Translation Reduction (\Cref{def:TransRed}).
Let $\Delta'(G)$ denote the average vertex degree of $G$.

\begin{restatable}{lemma}{simpleApproxLargerEven}
	\label{lemma:simpleApproxLargerEven}
	For every $k \in \mathbb N^+$, $\delta \in (\frac{1}{2k+2}, \frac{1}{2k+1})$ and graph $G$, $\covn[\frac{1}{2k+2}]{G} \leq {(1 + \frac{1}{k\Delta'(G) + 1})} \cdot \covn{G}$.
\end{restatable}
\begin{proof}
	By \cref{lemma:bound:delta:smaller:1}, for every $\delta \in (\frac{1}{2}, 1)$
	and non-tree connected graph $G$,
	$\covn[\half]{G} = |V(G)| \leq 2 \cdot \covn[\delta]{G}$.
	Consider a positive integer $k$.
	Applying the Translation Reduction $k$ times yields that a $\delta$-cover of $G$,
	for $\delta < \frac{1}{2k+1}$,
	has size at least $k|E(G)| + \frac{1}{2}|V(G)|$.
	On the other hand there is a $\frac{1}{2k+2}$-cover $S$ of $G$ with size $k|E(G)| + |V(G)|$
	consisting of a $V(G)$ and on every edge $k$ centered points with pairwise distance at least $2\delta$.
	Hence $S$ is an approximation with factor
	\[
	\frac{k|E(G)| + |V(G)|}{k|E(G)| + \frac{1}{2}|V(G)|} = 
	\frac{\left(k\frac{|E(G)|}{|V(G)|}+1\right)|V(G)|}{\left(k\frac{|E(G)|}{|V(G)|}+\frac{1}{2}\right)|V(G)|} = 
	\frac{\frac{k}{2}\Delta'(G)+1}{\frac{k}{2}\Delta'(G)+\frac{1}{2}} = 1+\frac{1}{k\Delta'(G)+1}.
	\]
\end{proof}

To derive bounds for the remaining intervals, $(\frac{1}{3},\half), (\frac{1}{5},\frac{1}{4})$ and so on,
	we recall the \Cref{equation:1:cover:size}:
$
	|S^*| \;=\; \nu(G_0) + \nu(G_{\geq 3}) + c_{\geq 3} + \text{vc}(G_1)
	\;\leq\; \tfrac{1}{2}(|V(G)| + c_{\geq 3})
	\;\leq\; \tfrac{2}{3}|V(G)|.
$
Further, we use that no component of $G_{\geq 3}$ is a tree, and thus $|E(G)| \geq |V(G)| - 1 + c_{\geq 3}$.

\begin{restatable}{lemma}{simpleApproxLargerOdd}
	\label{lemma:simpleApproxLargerOdd}
	Let $k \in \mathbb N^+$, $\delta \in (\frac{1}{2k+1}, \frac{1}{2k})$.
	For every $\varepsilon > 0$, there is an $n_0$ such that for every graph $G$
	on at least $n_0$ vertices,
	$\covn[\frac{1}{2k+1}]{G} \leq \min\{{1 + \frac{4}{3k\Delta'(G)}},\allowbreak {1 + \frac{1}{2k} + \varepsilon}\} \cdot \covn{G}$.
\end{restatable}
\begin{proof}
	Consider a graph $G$, integer $k \in \mathbb N^+$ and $\delta < \frac{1}{2k}$.
	We claim that a $\delta$-covering set $S$ has size at least $k|E(G)|$.
	Assuming otherwise, $S$ contains at most $k-1$ points from the interior of some edge $\{u,v\}$ of $G$.
	Let these points be $p(u,v,\lambda_1), \dots, p(u,v,\lambda_{k-1})$.
	We may assume that $\lambda_{i+1} - \lambda_{i} = 2\delta$ for all $i \in \{1,\dots,k-2\}$.
	Then, even if $S$ contains $u$ and $v$, the fraction of points in $P(G[\{u,v\}])$ that are covered is at 
	most 
	$2k\delta < 1$, which is a contradiction to $S$ being a $\delta$-cover.
	
	Further, by the Translation Reduction, there is a $\frac{1}{2k+1}$-cover of $G$
	of size at most $k |E(G)| + \covn[1]{G}$.
	Hence $\covn[\frac{1}{2k+1}]{G} \leq (1 + \frac{\covn[1]{G}}{k\cdot |E(G)|}) \cdot \covn[\delta]{G}$.
	Depending on the average degree $\Delta'(G)$, we bound this ratio as follows.
	First, if $\Delta'(G)$ is large, we use the above analysis that $\covn[1]{G} \leq \frac{2}{3}|V(G)|$ to 
	bounds 
	the factor to $1 + \frac{4}{3k\Delta'(G)}$.
	Otherwise, we combine the observation from \Cref{equation:1:cover:size} that $\covn[1]{G} \leq \frac{1}{2}(|V(G)| + c_{\geq 3})$ 
	with 
	$|E(G)| \geq |V(G)| + c_{\geq 3} - 1$ to bound the factor to $1 + \frac{1}{2k}\frac{|V(G)| + c_{\geq 
			3}}{|V(G)| + 
		c_{\geq 3} - 1}$.
	The latter follows from the fact that due to the components of $G_{\geq 3}$ being factor-critical, they cannot be trees.
	Since $\lim_{n \rightarrow \infty}\frac{n + c_{\geq 3}}{n + c_{\geq 3} - 1} = 1$, for every 
	$\varepsilon > 0$, 
	there is an $n_0$ such that the factor for any graph with at least $n_0$ vertices is at most $1 + 
	\frac{1}{2k} + 
	\varepsilon$.
\end{proof}

As there are only constant many graphs on less than $n_0$ vertices,
	we obtain an approximation of $\delta$-cover with approximation ratio
	$\min\{{1 + \frac{4}{3k\Delta'(G)}},\allowbreak {1 + \frac{1}{2k} + \varepsilon}\}$
	for any $\varepsilon >0$.

\newpage
\bibliographystyle{plain}
\bibliography{lit}

\end{document}